%% file: Draft_reviewpaper.tex
\documentclass[10pt]{article}
\expandafter\let\csname equation*\endcsname\relax
\expandafter\let\csname endequation*\endcsname\relax
\usepackage{bm}
\usepackage{pifont}
\usepackage{enumerate}
\usepackage[top=1in, bottom=1in, left=1in, right=1in]{geometry}
\usepackage[dvipsnames]{xcolor}
\usepackage{mathrsfs}
\usepackage{amsfonts}
\usepackage{multirow}
\usepackage{upgreek}
\usepackage{nicefrac}  
\usepackage{amssymb,dsfont}
\usepackage{caption,comment}
\usepackage{blkarray}
\usepackage{amsmath}
\usepackage{float}
\usepackage{footnote}
\usepackage{amssymb,amsthm}
\usepackage[toc,page]{appendix}
\usepackage{graphicx,afterpage}
\usepackage{epstopdf}
\usepackage[normalem]{ulem}
\usepackage{algorithm}
\usepackage{algorithmic}
\usepackage{xspace}
\usepackage{enumitem}
\usepackage{authblk}
\usepackage[pdftex,colorlinks=true,urlcolor=blue]{hyperref}
\usepackage{natbib}

\hypersetup{CJKbookmarks,%
	bookmarksnumbered,%
	colorlinks,%
	linkcolor=black,%
	citecolor=black,%
	plainpages=false,%
	pdfstartview=FitH}
\numberwithin{equation}{section}
\numberwithin{figure}{section}

\newcommand\tabcaption{\def\@captype{table}\caption}

\newtheorem{thm}{Theorem}[section]

\newtheorem{lemma}[thm]{Lemma}
\newtheorem{proposition}[thm]{Proposition}

\newcommand{\ssp}[1]{\langle #1 \rangle}
\newcommand{\norm}[1]{\|#1\|}
\newcommand{\Var}{\text{Var}}

\def\real{\mathbb{R}}

\newcommand{\cP}{\mathcal P}

\newcommand{\cF}{\mathcal F}

\renewcommand{\P}{\mathbb P}

\newcommand{\R}{\mathbb R}

\newcommand{\E}{\mathbb E}
\newcommand{\KL}{\mathrm{KL}}

\def\predictiveN{\mu^N_{n-1/2}}
\def\predictive{\mu_{n-1/2}}

\def\update{\mu_n}

\def\weight{v_n}
\def\weightN{w_n}
\newcommand{\supnorm}[1]{\norm{#1}_\infty}
\newcommand{\pconverges}{\ensuremath{\overset{p}{\rightarrow}}}

\newcommand{\testfn}{\ensuremath{\varphi} }
\newcommand{\bounded}[1][\real^d]{\ensuremath{\mathcal{B}_b(#1)}}

\def\mse{\mathrm{MSE}}
\def\FR{\textrm{FR}}
\def\W{\textrm{W}}

\usepackage{tikz}
\usetikzlibrary{positioning}
\title{Sequential Monte Carlo approximations \\of Wasserstein--Fisher--Rao gradient flows}
\date{}

\author[1]{Francesca Romana Crucinio\thanks{ \href{mailto:francescaromana.crucinio@unito.it}{francescaromana.crucinio@unito.it}}}
\author[2]{Sahani Pathiraja\thanks{  \href{mailto:s.pathiraja@unsw.edu.au}{s.pathiraja@unsw.edu.au}} }
\date{ }

\affil[1]{ESOMAS, University of Turin, Italy \& Collegio Carlo Alberto, Turin, Italy}
\affil[2]{School of Mathematics \& Statistics, UNSW Sydney, Australia}

\begin{document}

    \maketitle

\begin{abstract}
    We consider the problem of sampling from a probability distribution $\pi$. It is well known that this can be written as an optimisation problem over the space of probability distribution in which we aim to minimise the Kullback--Leibler divergence from $\pi$. We consider several partial differential equations (PDEs) whose solution is a minimiser of the Kullback--Leibler divergence from $\pi$ and connect them to well-known Monte Carlo algorithms. We focus in particular on PDEs obtained by considering the Wasserstein--Fisher--Rao geometry over the space of probabilities and show that these lead to a natural implementation using importance sampling and sequential Monte Carlo. We propose a novel algorithm to approximate the Wasserstein--Fisher--Rao flow of the Kullback--Leibler divergence and conduct an extensive empirical study to identify when this algorithm outperforms other popular Monte Carlo algorithms.
\end{abstract}


\input{paper1_body}

\paragraph*{Acknowledgements}
The authors would like to thank Deniz Akyildiz, Paula Cordero Encinar, Adrien Corenflos, Nikolas N\"usken, Anna Korba, Linfeng Wang for helpful discussions.  The authors gratefully acknowledge the mathematical research institute MATRIX in Australia where part of this research was performed.

F.R.C. gratefully acknowledges the ``de Castro" Statistics Initiative at the \textit{Collegio Carlo Alberto} and the \textit{Fondazione Franca e Diego de Castro}. F.R.C. is supported by the Gruppo
Nazionale per l'Analisi Matematica, la Probabilità e le loro Applicazioni (GNAMPA-INdAM).
S.P. gratefully acknowledges funding from UNSW Faculty of Science Research Grant and the Eva Mayr Stihl Foundation.

\bibliographystyle{plainnat}
\bibliography{mybibfile.bib}

\clearpage
\appendix
\numberwithin{equation}{section}
\numberwithin{figure}{section}

\input{paper1_appendix}
\end{document}

%% file: paper1_body.tex
\section{Introduction}

Sampling from a target probability distribution whose density is known up to a normalisation constant is a fundamental task in computational statistics and machine learning. A natural way to formulate this task is optimisation of a functional measuring the dissimilarity to the target probability distribution, typically the Kullback--Leibler (KL) divergence.

Following this point of view, one can derive many popular sampling frameworks including variational inference \citep{blei2017variational}, algorithms based on the Langevin dynamics \citep{roberts1996exponential, durmus2019analysis} and on kernelised Langevin dynamics \citep{liu2016stein}.
When the space over which the optimisation is carried out is the Wasserstein space, i.e. probability distributions with bounded second moments equipped with the Wasserstein-2 distance \citep{ambrosio2008gradient}, the resulting gradient descent scheme is referred to as Wasserstein (W) gradient flow. It is well-known that the W gradient flow of the KL can be implemented by a Langevin diffusion  \citep{jordan1998variational, wibisono2018sampling} and easily discretised in time, resulting for instance in the Unadjusted Langevin algorithm (ULA) and in the Metropolis-adjusted Langevin Algorithm (MALA).
Convergence of the W gradient flow and of ULA is guaranteed when the target distribution is sufficiently smooth and verifies a log-Sobolev assumption \citep{vempala2019rapid}.
Alternative discretisations of the KL Wasserstein gradient flow \citep{salim2020wasserstein, mou2021high} or its gradient flow with respect to similar optimal transport geometries have been considered in the literature to propose alternative algorithms
\citep{liu2016stein,garbuno2020interacting}, but their convergence also depends strongly on the log-concavity of the target.

A second class of gradient flows which has gained popularity in recent years is that of Fisher--Rao (FR) gradient flows.
Contrary to the Wasserstein metric, which leads to a diffusive behaviour, the Fisher--Rao metric leads to birth-death dynamics in which mass is created or removed. 
Exponential convergence of the FR gradient flow for KL does not require a log-Sobolev assumption and can be established under a bounded second moment assumption and mild regularity of the ratio between target and initial distribution \citep[Theorem 4.1 and Eq. (B.6)]{chen2023sampling}, and therefore holds for a wider class of targets. The FR gradient flow for KL is intimately tied to tempering or annealing  due to its unit time rescaling \citep{domingo-enrich2023an, chen_efficient_2024} and can also be connected to a mirror descent optimisation scheme for KL \citep{chopin2023connection}.
While superior in terms of theoretical behaviour, the FR gradient flow does not naturally lead to stable algorithms, in fact, as noted in \cite{chen_efficient_2024}, pure FR dynamics cannot change the support of the distribution, so additional steps need to be added to explore the space. These steps can change the dynamics and can become problematic in high dimensions. 

The main focus of this paper, is the combination of the W and the FR gradient flows obtained by considering the Wasserstein--Fisher--Rao metric \citep{gallouet2017jko} leading to the Wasserstein--Fisher--Rao (WFR) gradient flow. 
The WFR flow combines the diffusive behaviour of the W gradient flow with the FR flow and enjoys more favourable convergence properties, combining the rate of convergence of the W flow with that of the FR flow \citep{domingo-enrich2023an, Lu2019, us_bernoulli}; its numerical approximation can considerably speed up the convergence of gradient based methods for multimodal targets \citep{Lu2019, lu2023birth} and has been successfully employed to learn Gaussian mixtures \citep{yan2024learning}.
Despite its favourable convergence properties, the WFR flow of the KL has remained relatively understudied compared to the W and the FR flow. This is likely due to the difficulty of obtaining stable numerical approximations for the WFR flow. 


The main contributions of this work are:
\begin{itemize}
    \item We propose a stable algorithm to approximate the WFR flow which combines the diffusive behaviour of Langevin samplers and importance sampling (Algorithm~\ref{alg:smc_wfr}) and establish its convergence properties (Proposition~\ref{prop:is_wfr}, \ref{prop:wlln} and~\ref{prop:w2}). We show that this algorithm is a particular instance of sequential Monte Carlo samplers.
    \item We provide an optimisation point of view on sequential Monte Carlo samplers, showing that they provide numerical approximations of many FR type gradient flows of the Kullback--Leibler divergence. We thus establish a parallel result to that of \cite{jordan1998variational,wibisono2018sampling} which shows that algorithms based on the Langevin diffusion can be seen as numerical approximations of the Wasserstein gradient flow (Section~\ref{sec:smc}).  While connections between sequential Bayesian inference and FR flows have been explored in the discrete time and discrete space setting \citep{DelMoral1997,Shalizi2009,harper2009replicator,akyildiz2017probabilistic,Czegel2022} and more recently in continuous time and continuous space \citep{Pathiraja2024}; our focus here is on the sampling problem, in which no underlying sequential structure is present. 
    \item We empirically show (Section~\ref{sec:expe}) that  Algorithm~\ref{alg:smc_wfr} outperforms current approximations of the WFR flow and is more robust to initialisation than algorithm approximating the W and FR flows. We compare our approximation of the WFR flow with other popular Monte Carlo algorithms and identify regimes in which the addition of the FR correction to the W flow is beneficial. We conclude by providing guidelines for choosing the appropriate algorithm for a given target $\pi$.
\end{itemize}


\section{Wasserstein--Fisher--Rao gradient flow}
\label{sec:pde}

Let us denote the target distribution over $\R^d$ by $\pi= \exp(-V_\pi)$. A natural way to recover $\pi$ is to consider the optimisation problem
    $\min_{\mu \in \cP(\R^d)} \KL(\mu||\pi) = \min_{\mu \in \cP(\R^d)} \int \log\frac{\mu(x)}{\pi(x)}\mu(x)dx$,
which can be solved by following the direction of steepest descent of $ \KL(\cdot||\pi)$, given by the gradient flow partial differential equation (PDE).
In this paper, we focus on the gradient flow of $\KL$ w.r.t. the Wasserstein--Fisher--Rao (WFR) geometry.
The WFR PDE is 
\begin{align}
\label{eq:wfr}
    \partial_t \mu_t = f_{\W}(\mu_t)+f_{\FR}(\mu_t) \quad t \in [0,\infty),
\end{align}
where $f_{\W}(\mu):= \nabla \cdot (\mu \nabla \log \frac{\mu}{\pi})$ and $f_{\FR}(\mu) := \mu (\log \frac{\pi}{\mu} -  \int \log \frac{\pi}{\mu} d\mu  )$ correspond to the Wasserstein and Fisher--Rao operators respectively, and takes the form of a reaction-diffusion equation. It combines the diffusive behaviour of the Fokker-Planck equation of a Langevin diffusion targeting $\pi$ \citep{jordan1998variational} which promotes exploration, with the tempering mechanism typical of geometric interpolations or annealing \citep{Lu2019, chopin2023connection, domingo-enrich2023an}.

We briefly review the convergence properties of the W, FR and WFR flows to showcase the superior theoretical properties of the latter.
Under a log-Sobolev assumption on the target $\pi$
\begin{align}
    \label{eq:lsi}
    \KL(\mu||\pi)\leq \frac{C_{\textrm{LSI}}}{2}\int \mu \Vert\nabla\log \frac{\mu}{\pi}\Vert^2,
\end{align}
where $C_{\textrm{LSI}}$ denotes the log-Sobolev constant of $\pi$, it is well known that (e.g. \citet[Section 2.1]{chewi2024analysis})
\begin{align*}
    \KL(\mu_t^{\textrm{W}}||\pi)\leq e^{-2C_{\textrm{LSI}}^{-1}t}\KL(\mu_0||\pi)
\end{align*}
where $\mu_t^{\textrm{W}}$ is the solution of $\partial_t\mu_t = f_\W(\mu_t)$.
For the FR flow, to the best of our knowledge, the convergence result obtained under the weakest assumptions on $\pi, \mu_0$ is given in \citet[Theorem 4.1]{chen2023sampling} and requires that $\mu_0, \pi$ have bounded second moments $\int \norm{x}^2\mu_0(x)dx \leq B$, $\int \norm{x}^2\pi(x)dx \leq B$
 and $\left|\log\nicefrac{\mu_0(x)}{\pi(x)}\right|\leq M(1+\norm{x}^2)$
for some $M>0$. Under these assumptions, we have
\begin{align*}
\KL(\mu_t^{\textrm{FR}}||\pi)\leq M(2+B+Be)e^{-t},
\end{align*}
for $t \geq \log((1+B)M)$ where $\mu_t^{\FR} $ is the solution of $\partial_t\mu_t= f_\FR(\mu_t)$.
Due to the orthogonality of the W and FR geometries \citep[page 13]{gallouet2017jko}, the convergence rate of WFR gradient flow~\eqref{eq:wfr} can be upper bounded by \citep[Remark 2.6]{lu2023birth}
\begin{align}
\label{eq:rate_wfr}
\KL(\mu_t^{\textrm{WFR}}||\pi)\leq \min\left\lbrace \KL(\mu_t^{\textrm{FR}}|\pi), \KL(\mu_t^{\textrm{W}}|\pi)\right\rbrace.
\end{align}

\begin{figure}
    \centering
    \includegraphics[width=0.45\linewidth]{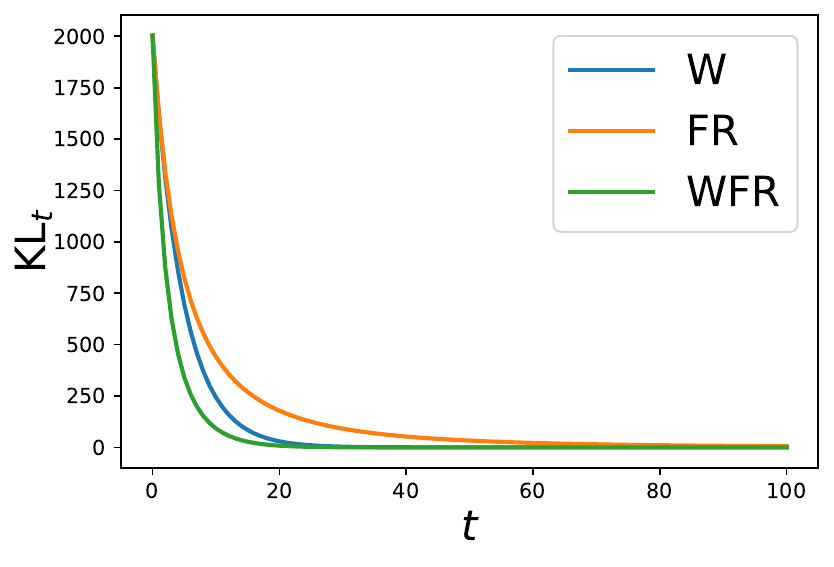}
    \includegraphics[width=0.45\linewidth]{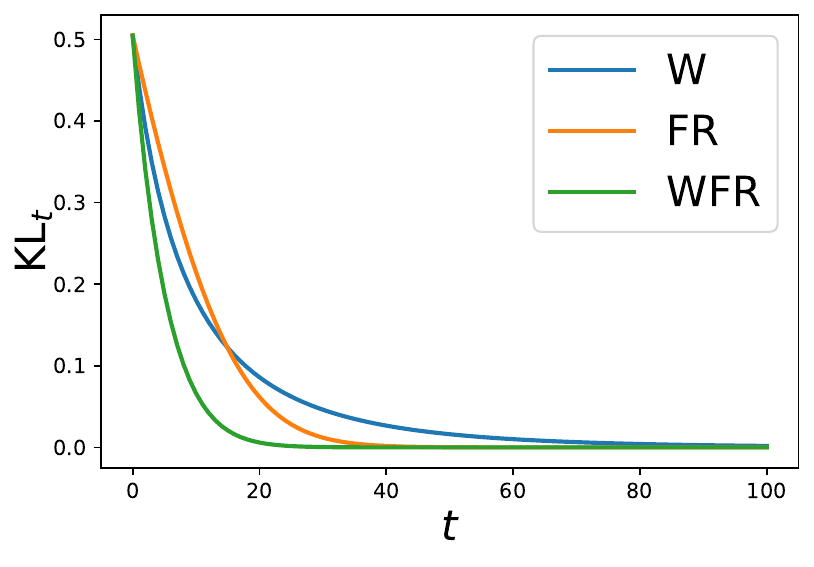}
    \caption{Evolution of $\KL$ along different PDE flows in the 1D Gaussian case with $\mu_0(x) = \mathcal{N}(x; 0, 1)$ and $\pi(x) = \mathcal{N}(x; 20, 0.1)$ (first row), $\pi(x) = \mathcal{N}(x; 1, 5)$ (second row).}
    \label{fig:gaussians_exact}
\end{figure}
Figure~\ref{fig:gaussians_exact} shows the evolution of Kullback--Leibler divergence along different PDE flows in the case of 1D Gaussians. The initial distribution is a standard Gaussian and we consider two different targets: for target 1 (first row) we set $m_\pi = 20, C_\pi = 0.1$.
In this case the W flow is clearly faster than the FR one and the WFR flow behaves more like the W flow (Figure~\ref{fig:gaussians_exact} left).
For the second target (right) we have $m_\pi = 1, C_\pi = 5$;
the FR flow is faster than the W flow for all sufficiently large $t$, and the WFR flow more closely follows the FR flow.
Figure~\ref{fig:gaussians_exact} confirms that combining the W flow with the FR flow improves on both the W flow and the FR flow.

\section{Numerical approximation of the WFR PDE}
\label{sec:smc_wfr}

We now consider numerical approximations of the WFR flow~\eqref{eq:wfr}. Current approximations of the WFR proceed by first discretising space by replacing $\mu_t$ with a particle approximation $\mu_t^{N}=N^{-1}\sum_{i=1}^N\delta_{X_t^i}$ with $X_t^i\sim \mu_t$ and secondly discretising time (\cite{Lu2019, lu2023birth} and \citet[Section H]{ lambert2022variational} for Gaussian distributions).
\cite{Lu2019, lu2023birth} consider $X_t^i$ for $i=1, \dots, N$ which follows a Langevin diffusion to which a birth-death step is added to mimic the FR dynamics. \citet[Section H]{ lambert2022variational} considers weighted particle approximations $\mu_t^{N}=N^{-1}\sum_{i=1}^NW_t^i \delta_{X_t^i}$ where $X_t^i$ is a Gaussian approximation to $\mu_t$ and $W_t$ follows an ODE derived from the FR flow.

We consider a different approach: we first discretise time and then replace $\mu_n$ (a time discretised version of $\mu_t$) with a weighted particle approximation  $\mu_n^{N}=\sum_{i=1}^NW_n^i\delta_{X_n^i}$. Considering a time discretisation before a space discretisation allows us to exploit not only the connections between the W flow and Langevin dynamics but also those between the FR flow and mirror descent dynamics \citep{chopin2023connection}.

\subsection{Time discretisation}
\label{sec:3.1}

To obtain a time discretisation of the WFR PDE we consider the solution operator of the W flow and the FR flow separately. Assume that for certain discrete time $n-1$ the solution of the WFR PDE is given by $\mu_{n-1}$.
A discretisation of the W flow for a time step $\gamma$ leads to
\begin{align}
\label{eq:w_semigroup}
     \mu_{n-1/2}(x) &  := \mathcal{N}(0, 2\gamma\cdot \textsf{Id})*\left[(\textsf{Id}+\gamma\nabla\log \pi)_{\#}\mu_{n-1}\right](x) \\
     & \propto \int_{\R^d} \exp \left(-\frac{1}{4\gamma} \|x - y - \gamma  \nabla \log \pi(y) \|^2   \right) \mu_{n-1}(y) dy, \notag
\end{align}
see, e.g. \cite{wibisono2018sampling}. 
The FR flow with time step $\gamma$ applied to $\mu_{n-1/2}$ gives (e.g. \citet[App. B.1]{chen2023sampling}, \citet[Eq. (16)]{lu2023birth})
\begin{align}
\label{eq:fr_semigroup}
   \mu_n(x) \propto \pi(x)^{1 - e^{-\gamma}} \mu_{n-1/2}(x)^{e^{-\gamma}} = \left( \frac{\pi(x)}{\mu_{n-1/2}(x)} \right)^{1-e^{-\gamma}}  \mu_{n-1/2}(x).
\end{align}
The solution of the FR flow~\eqref{eq:fr_semigroup} corresponds to one step of mirror descent applied to the Kullback--Leibler divergence (see \cite{chopin2023connection}).
Alternating between one step of the W solution operator~\eqref{eq:w_semigroup} and one step of the FR one~\eqref{eq:fr_semigroup} we obtain a discrete-time splitting scheme approximating of WFR PDE.
It is known that this splitting scheme converges to the WFR PDE uniformly in time at rate $\mathcal{O}(\gamma)$ as $\gamma \rightarrow 0$ under mild conditions on $f_\textrm{FR}$, $f_\textrm{W}$ and $\pi$ (e.g. \cite{Hundsdorfer2003}).  

Combining the fact that the mirror descent iterates reduce the Kullback--Leibler divergence as established in \cite{aubin2022mirror, chopin2023connection} and standard results on the convergence of the unadjusted Langevin algorithm \citep{vempala2019rapid}, we can establish that the discrete scheme given by~\eqref{eq:w_semigroup}--\eqref{eq:fr_semigroup} converges exponentially fast (see Appendix~\ref{app:proofs3} for a proof).

\begin{proposition}
\label{prop:is_wfr}
Assume that $\pi$ satisfies the log-Sobolev inequality~\eqref{eq:lsi},  that there exists $L_\pi>0$ such that $\|\nabla V_\pi(x) - \nabla V_\pi(x')\| \leq L_\pi\|x-x'\|$ and that $\KL(\mu_0||\pi)<\infty$. If $0< \gamma \leq C_{\textrm{LSI}}^{-1}/(4L_\pi^2)$ the approximation of the WFR flow given by~\eqref{eq:w_semigroup}--\eqref{eq:fr_semigroup} satisfies
$$\KL(\mu_n||\pi)\leq \min\left\lbrace\phi_{\textrm{W}}(\gamma, n, \mu_0, C_\textrm{LSI}),\phi_{\textrm{FR}}(\gamma, n,\mu_0)\right\rbrace$$
where
\begin{align}
\label{eq:w_rate}
\phi_{\textrm{W}}(\gamma, n, \mu_0, C_\textrm{LSI}) &:= e^{-n\gamma/C_\textrm{LSI}}\KL(\mu_0||\pi)
+
6\gamma^2 d L_\pi^2
\frac{1-e^{-n\gamma/C_\textrm{LSI}}}
{1-e^{-\gamma/C_\textrm{LSI}}}
\\
&\quad
-\frac{1}{1-e^{-\gamma}}
\sum_{k=1}^n
e^{-(n-k)\gamma/C_\textrm{LSI}}
\int_{e^{-\gamma}}^1
(s-e^{-\gamma})
\operatorname{Var}_{\nu_s}
\!\left(
\log\frac{\mu_{k-1/2}}{\pi}
\right)
\,ds.\notag\\
    \label{eq:fr_rate}
    \phi_{\textrm{FR}}(\gamma, n,\mu_0)
    &:=\frac{1-e^{-n\gamma}}{1-e^{-\gamma}}
\KL(\rho_\gamma||\pi)
+
\frac{e^{-n\gamma}}{1-e^{-\gamma}}
\KL(\rho_\gamma||\mu_0),
\end{align}
and $\rho_\gamma$ denotes the invariant measure of the ULA kernel
\begin{align}
    \label{eq:Rgamma}
R_{\gamma}(x, dy)=(4\uppi\gamma)^{-d/2} \exp \left(-\frac{1}{4\gamma} \|y - x -\gamma\nabla \log \pi(x) \|^2   \right) dy
\end{align}
and $\KL(\rho_\gamma||\pi)\leq 8L_\pi^2\gamma d C_\textrm{LSI}$ by \citet[Corollary 1]{vempala2019rapid}.
\end{proposition}

This result shows that also in discrete time WFR achieves convergence speed which is the minimum between the W rate and the FR rate: \eqref{eq:fr_rate} corresponds to the mirror descent (i.e. discrete time FR) rate of convergence to $\rho_\gamma$ plus a term taking into account the discretisation bias; \eqref{eq:w_rate} corresponds to the classical convergence rate of ULA \citep[Theorem 1]{vempala2019rapid} minus a positive quantity due to the presence of the FR flow.  Since as $\gamma\to 0, n\to \infty, n\gamma \to \infty$ we have $\phi_\W\to 0$ we can immediately conclude $\KL(\mu_n||\pi)\to 0$.
The result in Proposition~\ref{prop:is_wfr} matches the min rate in~\eqref{eq:rate_wfr}, we expect that to obtain sharper rates of convergence one needs to take into account the ordering in which the W and FR flows are applied as shown in \cite{us_splitting}.


\subsection{A sampling approach to approximate the WFR PDE}
\label{sec:proxy}
The evolution in~\eqref{eq:w_semigroup} and~\eqref{eq:fr_semigroup} describes one step of the WFR in discrete time. A discretisation of~\eqref{eq:w_semigroup} over timestep $\gamma$ is given by the unadjusted Langevin algorithm (ULA) \citep{jordan1998variational}. In particular, we can approximate $\mu_{n-1/2}$ using $N$ particles $X_{n-1}^1, \dots, X_{n-1}^N$ initialised at $\mu_{n-1}$
\begin{align}
\label{eq:ula}
    X_{n}^i = X_{n-1}^i + \gamma\nabla\log\pi(X_{n-1}^i)+\sqrt{2\gamma}\xi_{n}^i,
\end{align}
where $\xi_{n}^i\sim \mathcal{N}(0, \textsf{Id})$ for $i=1, \dots, N$. The empirical measure $\mu_{n-1/2}^N :=N^{-1}\sum_{i=1}^N \delta_{X_{n}^i} $ provides an approximation to $\mu_{n-1/2}$.
The expression \eqref{eq:fr_semigroup} shows that given $\mu_{n-1/2}^N$, $\mu_{n}$ can be approximated by importance sampling with unnormalised weights $w_n(x) = ( \pi(x)/\mu_{n-1/2}^N(x) )^{1-e^{-\gamma}}$ (notice that $w_n$ requires density estimation techniques, we return on this later in this section). These weights approximate the weights of an idealised importance sampler in which $\mu_{n-1/2}$ can be computed analytically
\begin{align}
\label{eq:weights_ideal}
    v_n(x) &=\left( \frac{\pi(x)}{\mu_{n-1/2}(x)}\right)^{1-e^{-\gamma}} = \left( \frac{\pi(x)}{\int_{\R^d}R_\gamma(y, x)\mu_{n-1}(y)}\right)^{1-e^{-\gamma}} 
\end{align}
with $R_\gamma$ as in~\eqref{eq:Rgamma}.
Alternating between one step of W flow approximated via~\eqref{eq:ula} and one step of FR flow approximated via importance sampling we obtain an approximation of WFR PDE through sampling-based procedures. To avoid the weight degeneracy normally introduced by repeated importance sampling steps, after each FR step we add a resampling step \citep{jasra2011inference}.
We thus obtain two approximations of $\mu_n$, one after the importance sampling step $\mu_n^N:=\sum_{i=1}^N W_n^i\delta_{X_n^i}$ where $W_n^i\propto w_n(X_n^i)$ and one after the resampling step $\widetilde{\mu}_n^N:=N^{-1}\sum_{i=1}^N \delta_{\widetilde{X}_n^i}$ where $\widetilde{X}_n^i$ are sampled from the categorical distribution $\{X_n^j, W_n^j\}_{j=1}^N$.

Several density approximation techniques can be employed to approximate $w_n$. We propose to exploit the fact that each FR step follows a W step to obtain an approximation $q_n^N$ of the density of $\mu_{n-1/2}$
\begin{align}
    \label{eq:qN}
    q_n^N(x) &:= \mathcal{N}(0, 2\gamma\cdot\textsf{Id})*\left[(\textsf{Id}+\gamma\nabla\log \pi)_{\#}\mu_{n-1}^N\right](x)\\
    &=\sum_{i=1}^NW_{n-1}^i\mathcal{N}(x; X_{n-1}^i+\gamma\nabla\log\pi(X_{n-1}^i), 2\gamma\cdot\textsf{Id})\notag
\end{align}
where we used the definition of pushforward measure and convolution integral to obtain the last expression,
and define 
\begin{align}
\label{eq:w_wfr}
    w_n(x) &=\left( \frac{\pi(x)}{q_{n}^N(x)}\right)^{1-e^{-\gamma}}.
\end{align}
The resulting procedure is summarised in Algorithm~\ref{alg:smc_wfr}. We show in Section~\ref{sec:smc} that this approximation of the WFR PDE can be linked to a well known class of Monte Carlo algorithms: sequential Monte Carlo samplers.

Similarly to the approximations of mirror descent \citep{dai2016provable, korba2022adaptive, bianchi2024stochastic} and of birth-death Langevin algorithms \citep{Lu2019, lu2023birth}, the cost of computing the approximate weights is linear in $N$. However, in our case the mixture approximating the denominator is not obtained from a kernel density estimator introduced purely for numerical purposes, but can be justified from the point of view of the W flow. We directly compare this approximation of the WFR flow with that provided by the birth-death Langevin dynamics in Section~\ref{sec:expe}.

\begin{algorithm}
\begin{algorithmic}[1]
\STATE{\textit{Inputs:} learning rate $\gamma$, initial proposal $\mu_0$.}
\STATE{\textit{Initialize:} sample $\widetilde{X}_0^i\sim \mu_0$ and set $W_0^i=1/N$ for $i=1,\dots, N$.}
\FOR{$n=1,\dots, T$}
\IF{$n>1$}
\STATE{\textit{Resample:} draw $ \{\widetilde{X}_{n-1}^i\}_{i=1}^N$ independently from $\{X_{n-1}^i, W_{n-1}^i\}_{i=1}^N$  for $i=1,\dots, N$.}
\ENDIF
\STATE{\textit{W flow:} update $X_{n}^i =\widetilde{X}_{n-1}^i +\gamma\nabla\log \pi(\widetilde{X}_{n-1}^i)+\sqrt{2\gamma}\xi_{n-1/2}^i$, where $\xi_{n-1/2}^i\sim \mathcal{N}(0, \textsf{Id})$ for $i=1,\dots, N$.}
\STATE{\textit{FR flow:} compute and normalize the weights $W_n^i \propto w_{n}(X_{n}^i)$ for $i=1,\dots, N$, where
\begin{align*}
    w_{n}(x) = \left(\frac{\pi(x)}{\sum_{i=1}^NW_{n-1}^i\mathcal{N}(x;X_{n-1}^i+\gamma\nabla\log\pi(X_{n-1}^i), 2\gamma\cdot \textsf{Id})}\right)^{1-e^{-\gamma}} 
\end{align*}} 
\ENDFOR
\STATE{\textit{Output:} $\{ X_n^i, W_n^i\}_{i=1}^N$}
\end{algorithmic}
\caption{SMC-WFR.}\label{alg:smc_wfr}
\end{algorithm}

Convergence of the numerical approximations provided by Algorithm~\ref{alg:smc_wfr} follows using similar arguments to those usually employed in the sequential Monte Carlo (SMC) literature (see Section~\ref{sec:smc} for the connection between Algorithm~\ref{alg:smc_wfr} and these methods). 
We show convergence for the approximations of $\mu_n$ provided at each step of Algorithm~\ref{alg:smc_wfr}, the proof is provided in Appendix~\ref{app:lp}. For simplicity, we focus on converge for integrals of measurable bounded test functions $\testfn:\real^d\to\real$ with $\supnorm{\testfn}:= \sup_{x\in  \real^d}\vert \testfn(x)\vert<\infty$, a set we denote by $\bounded$. 
In our arguments we only consider multinomial resampling \citep{smc:methodology:GSS93} although lower variance resampling schemes exist \citep{Gerber2019}.


\begin{proposition}[Consistency of SMC-WFR]
\label{prop:wlln}
Assume that
\begin{enumerate}[label=(\roman*)]
\item there exists $L_\pi>0$ such that $\|\nabla V_\pi(x) - \nabla V_\pi(x')\| \leq L_\pi\|x-x'\|$;
\item for any $\gamma>0$, $\int_{\R^d}\pi(x)^{1-e^{-\gamma}}dx<\infty$.
\end{enumerate}
Furthermore, assume that, for every
$n\geq 1$,
\begin{equation}
\label{ass:B}
\frac{1}{N}
\int_{\mathbb R^d}
w_n(x)^{2}
q_{n}^N(x)dx\pconverges0\qquad \textrm{as } N\to \infty.
\end{equation}
Then, for every 
$\varphi \in \bounded$,
\(
\sum_{i=1}^N
W_n^i\varphi(X_n^i) = \int_{\R^d}\varphi(x)\mu_n^N(dx)
\pconverges
\int_{\R^d}\varphi(x)\mu_n(x)dx
\) as $N\to \infty$.
Therefore, $\mu_n^N$ converges in the weak topology 
to $\mu_n$ in probability.
\end{proposition}

Weak convergence in probability is obtained assuming that \textit{(i)} the ULA step is stable \textit{(ii)} some control on the tails of the importance sampling step approximating the FR flow. Condition~\eqref{ass:B} is satisfied if, for example, $\norm{w_n}_{L^2(q_n^N)}\leq C$, which is substantially weaker than assuming bounded weights, $\supnorm{w_n}\leq C$, a standard assumption in the SMC literature. In particular, \eqref{ass:B} only requires averaged $L^2(q_n^N)$ control and therefore allows large weights on sufficiently rare regions.
Under additional assumptions which allow control of the second moment of $\mu_n, \mu_n^N$ we obtain convergence in Wasserstein-2 distance in probability, the proof is provided in Appendix~\ref{app:lp}.
\begin{proposition}[$W_2$-convergence]
\label{prop:w2}
    Under the assumptions of Proposition~\ref{prop:wlln}, further assume
\begin{enumerate}[label=(\roman*)]
\item $\pi(x) \propto e^{-V_\pi(x)}$ is $\alpha_\pi$-strongly log-concave, i.e. there exists an $L_\pi\geq \alpha_\pi > 0$ such that $ L_\pi I \succeq \nabla^2 V_\pi(x) \succeq \alpha_\pi I$ for all $x\in\R^d$;
\item there exists $B_0>0$ such that $\int_{\R^d}\norm{x}^2\mu_0(x)dx\leq B_0$.
\end{enumerate}
Then, $ \mu_n, \mu_n^N \in \cP_2(\R^d)$ almost surely for all $n\geq 0$ and $W_2(\mu_n^N, \mu_n)\pconverges 0$ as $N\to\infty$.
\end{proposition}

Propositions~\ref{prop:wlln} and~\ref{prop:w2} show that the particle system given by Algorithm~\ref{alg:smc_wfr} converges to the discrete time evolution in Section~\ref{sec:3.1} as $N\to \infty$ for each $n\geq 0$ and fixed $\gamma>0$. Propositions~\ref{prop:w2} and~\ref{prop:is_wfr} taken together imply that the particle system given by Algorithm~\ref{alg:smc_wfr} converges to $\pi$ when the limit in $N$ is taken first.   
This does not immediately give convergence of $\mu_n^N$ to $\pi$ in the joint $\gamma\to 0, n\gamma \to \infty, N\to \infty$ limit; we conjecture that to obtain such a limit one would need to consider a discretisation step $\gamma$ which decays with $N$. Nevertheless, our numerical experiments in Figure~\ref{fig:gaussians_exact_smc_variants} suggest that SMC-WFR does indeed approximate the WFR flow.

\section{Sequential Monte Carlo as Approximation of WFR Flows}
\label{sec:smc}

We introduce in this section a class of algorithms known as sequential Monte Carlo samplers \citep{del2006sequential} and highlight how different algorithmic choices within this class lead to approximations of the PDE flows described in Section~\ref{sec:pde}, including Algorithm~\ref{alg:smc_wfr}.


Sequential Monte Carlo (SMC) samplers \citep{del2006sequential} provide a particle approximation of the sequence of distributions $(\hat{\eta}_n)_{n\geq 0}$ using clouds of $N$ weighted particles $\{X_n^i, W_n^i\}_{i=1}^N$. 
To build an SMC sampler approximating the sequence $(\hat{\eta}_n)_{n\geq 0}$ one needs a family of Markov kernels $(M_n)_{n\geq 1}$ used to propagate forward the particles, a sequence of backward kernels $(L_{n-1})_{n \geq 1}$ and a resampling scheme.
A resampling scheme is a selection mechanism which given a set of weighted samples $\lbrace X^i, W^i\rbrace_{i=1}^N$ outputs a sequence of equally weighted samples $\lbrace \widetilde{X}^i, 1/N\rbrace_{i=1}^N$ in which $\widetilde{X}^i = X^j$ for some $j$ for all $i=1,\ldots, N$. For a review of commonly used resampling schemes see \citep{Gerber2019}.

At iteration $n$ the weighted particle set $\{X_{n-1}^i, W_{n-1}^i\}_{i=1}^N$ approximating $\hat{\eta}_{n-1}$ is resampled to obtain the equally weighted particle set $\{\widetilde{X}_{n-1}^i, 1/N\}_{i=1}^N$ and the kernel $M_{n}$ is applied to propose new particle locations $X_{n}^i\sim M_{n}(\widetilde{X}_{n-1}^i, \cdot)$. At this stage the particle set approximates $\hat{\eta}_{n-1}$, to obtain an approximation of $\hat{\eta}_n$ an importance sampling step is applied with weights proportional to
\begin{align}
\label{eq:smc_generalw}
    w_{n}(x_{n-1}, x_n) = \frac{\hat{\eta}_n(x_n)L_{n-1}(x_n, x_{n-1})}{\hat{\eta}_{n-1}(x_{n-1})M_{n}(x_{n-1}, x_{n})}.
\end{align}

The resulting procedure is summarised in Algorithm~\ref{alg:smc} whose output is a set of weighted samples which approximate $\hat{\eta}_n$. We refer the reader to \cite{chopin2020introduction} for a comprehensive introduction to SMC. Annealed importance sampling can be obtained from Algorithm~\ref{alg:smc} by omitting the resampling step.


\begin{algorithm}
\begin{algorithmic}[1]
\STATE{\textit{Inputs:} sequences of  distributions $(\hat{\eta}_n)_{n\geq 0}$, Markov kernels $(M_n)_{n\geq1}$, backward kernels $(L_{n-1})_{n\geq1}$.}
\STATE{\textit{Initialize:} sample $\widetilde{X}_0^i\sim \hat{\eta}_0$ and set $W_0^i=1/N$ for $i=1,\dots, N$.}
\FOR{$n=1,\dots, T$}
\IF{$n>1$}
\STATE{\textit{Resample:} draw $ \{\widetilde{X}_{n-1}^i\}_{i=1}^N$ independently from $\{X_{n-1}^i, W_{n-1}^i\}_{i=1}^N$ for $i=1,\dots, N$.}
\ENDIF
\STATE{\textit{Propose:} draw $X_n^i\sim M_{n}(\widetilde{X}_{n-1}^i, \cdot)$ for $i=1,\dots, N$.}
\STATE{\textit{Reweight:} compute and normalize the weights $W_n^i \propto w_{n}(\widetilde{X}_{n-1}^i, X_n^i)$ in~\eqref{eq:smc_generalw} for $i=1,\dots, N$.
}
\ENDFOR
\STATE{\textit{Output:} $\{X_n^i, W_n^i\}_{i=1}^N$}
\end{algorithmic}
\caption{SMC samplers \citep{del2006sequential}.}\label{alg:smc}
\end{algorithm}

Several choices of kernels $M_n, L_n$ are possible. We refer to \citet[Section 3.3]{del2006sequential} for a thorough discussion but mention here some options which will be relevant in the following sections.
A popular choice in the literature is to pick $M_n$ to be an MCMC kernel of invariant distribution $\hat{\eta}_n$. In this case, choosing
\begin{align}
\label{eq:backward}
    L_{n-1}(x_n, x_{n-1}) = \frac{ \hat{\eta}_n(x_{n-1}) M_n(x_{n-1}, x_n)}{\hat{\eta}_n(x_n)},
\end{align}
i.e. the reversed Markov kernel that is associated with $M_n$ allows to obtain a simplified expression for the weights $w_n(x_{n-1}) = \hat{\eta}_n(x_{n-1})/\hat{\eta}_{n-1}(x_{n-1})$.

\citet[Proposition 1]{del2006sequential} derives the expression of the optimal backward kernel $L_{n-1}$ which minimises the variance of the weights. This kernel is intractable in most cases; a pragmatic but suboptimal choice is given by the following approximation of the optimal kernel
\begin{align}
\label{eq:optimal_L}
L_{n-1}(x_n, x_{n-1}) = \frac{\hat{\eta}_{n-1}(x_{n-1})M_n(x_{n-1}, x_n)}{\int \hat{\eta}_{n-1}(z)M_n(z, x_n)dz}.
\end{align}
The resulting weight expression is
\begin{align}
\label{eq:w_backward_opt}
    w_n(x_n)&= \frac{\hat{\eta}_{n}(x_{n})}{\int \hat{\eta}_{n-1}(z)M_n(z, x_n)dz}.
\end{align}
There is a vast literature on choosing the sequence $(\hat{\eta}_n)_{\geq 0}$. The following section discusses how different choices lead to numerical approximation of the WFR PDE described above.

\subsection{SMC tempering}
\label{sec:smc_fr}
As a building block towards building approximations of the WFR PDE~\eqref{eq:wfr} via SMC, we consider the classical SMC tempering algorithm (see, e.g., \citet[Chapter 17]{chopin2020introduction}) and highlight its connections with the Fisher--Rao PDE.
As discussed in Section~\ref{sec:smc_wfr}, the discrete time solution operator of the FR flow~\eqref{eq:fr_semigroup} applied to $\mu_{n-1}$ is given by $\mu_n(x) \propto \pi(x)^{1-e^{-\gamma}}\mu_{n-1}(x)^{e^{-\gamma}}$.
Using the recursive structure of~\eqref{eq:fr_semigroup} we have
\begin{align}
\label{eq:tempering}
\mu_n(x) \propto \pi(x)^{1-e^{-\gamma}}\mu_{n-1}(x)^{e^{-\gamma}} \propto \pi(x)^{1-e^{-2\gamma}}\mu_{n-2}(x)^{e^{-2\gamma}}\propto\dots \propto \mu_{0}(x)^{e^{-n\gamma}}\pi(x)^{1-e^{-n\gamma}},
\end{align}
%
i.e. the tempering iterates popular in the SMC literature \citep{chopin2023connection}.

Algorithm~\ref{alg:smc} with $\hat{\eta}$ equal to $\mu_n$ in~\eqref{eq:tempering}, Markov kernel $M_n$ a simple random walk Metropolis (RWM) kernel leaving $\hat{\eta}_{n}=\mu_n$ invariant and backward kernel $L_{n-1}$ as in~\eqref{eq:backward} is typically known as SMC tempering. Due to the structure of~\eqref{eq:tempering} the weights simplify to 
\begin{align}
\label{eq:weightsSMC tempering}
    w_n(x_{n-1}) = \frac{\hat{\eta}_n(x_{n-1})}{\hat{\eta}_{n-1}(x_{n-1})}=\frac{\mu_n(x_{n-1})}{\mu_{n-1}(x_{n-1})}=\left(\frac{\pi(x_{n-1})}{\mu_0(x_{n-1})}\right)^{e^{-(n-1)\gamma}-e^{-n\gamma}}
\end{align}
and each weight can be computed in $\mathcal{O}(1)$ w.r.t. $N$.
The resulting SMC algorithm approximates the FR flow;
Figure~\ref{fig:gaussians_exact_smc} first row compares the approximation provided by Algorithm~\ref{alg:smc} and the evolution of the FR flow for a 1D Gaussian distribution $\pi(x)=\mathcal{N}(x; 20, 0.1)$.

\begin{figure}
    \centering
    \includegraphics[width=0.8\linewidth]{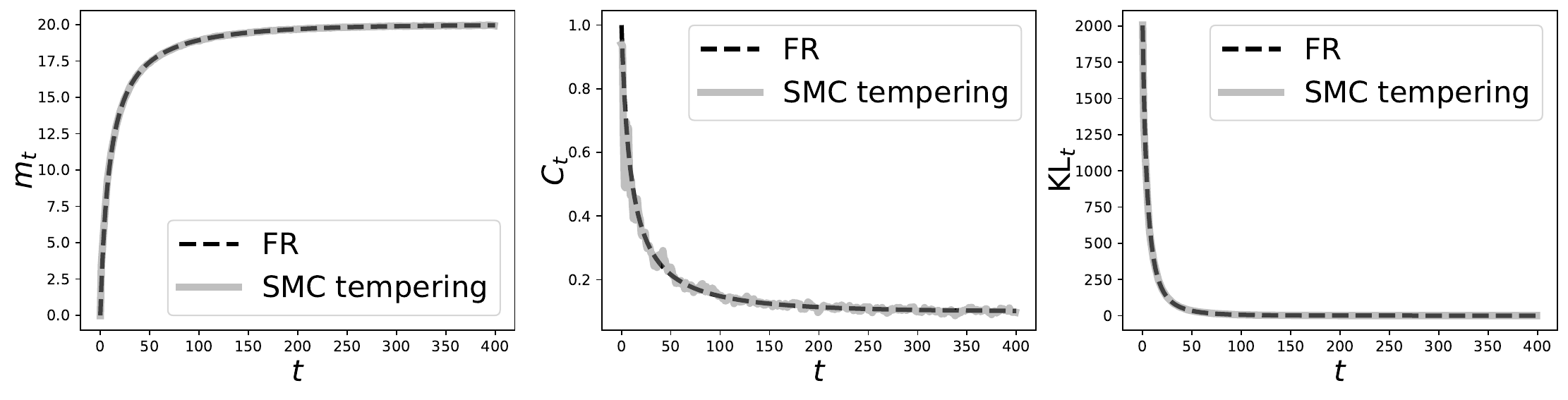}
    \caption{Comparison of evolution of mean, variance and $\KL$ of the exact FR flow and approximation provided by SMC tempering with target distribution the 1D Gaussian $\pi(x) = \mathcal{N}(x; 20, 0.1)$ and initial distribution $\mu_0(x) = \mathcal{N}(x; 0, 1)$.}
    \label{fig:gaussians_exact_smc}
\end{figure}

\subsection{Wasserstein--Fisher--Rao flow through SMC}

We saw in the previous section that classical SMC tempering provides an approximation to the discrete time FR flow~\eqref{eq:tempering}.
We now identify the sequence of distributions $(\hat{\eta}_n)_{n\geq0}$, the Markov kernel $M_n$ and backward kernel $L_{n-1}$ which lead to SMC-WFR. Let $ \hat{\eta}_0\equiv 
\mu_0$ and
\begin{align}
\label{eq:new_seq}
    \hat{\eta}_n(x)\propto \left(\hat{\eta}_{n-1}R_\gamma(x)\right)^{e^{-
    \gamma}}\pi(x)^{1-e^{-
    \gamma}}.
\end{align}
Selecting as proposal kernel $M_n\equiv R_\gamma$ and the optimal backward kernel \eqref{eq:optimal_L} leads to 
\begin{align*}
    w_n(x_n) = \frac{\left(\hat{\eta}_{n-1}R_\gamma(x_n)\right)^{e^{-
    \gamma}}\pi(x_n)^{1-e^{-
    \gamma}}}{\hat{\eta}_{n-1}R_\gamma(x_n)} =  \left(\frac{\pi(x_n)}{\hat{\eta}_{n-1}R_\gamma (x_{n})}\right)^{1-e^{-\gamma}}.
\end{align*}
Plugging $(\hat{\eta}_n)_{n\geq0}$, $M_n$ and the weights above into Algorithm \ref{alg:smc} and using the approximation $\hat{\eta}_{n-1}R_\gamma (x_{n})\approx \sum_{i=1}^N W_{n-1}^i R_\gamma(X_{n-1}^i, x_n)$ we obtain SMC-WFR, showing that Algorithm \ref{alg:smc_wfr} is indeed an instance of SMC.

The sequence of distributions \eqref{eq:new_seq} does not correspond to the tempering iterates \eqref{eq:tempering} obtained by discretising the FR flow. 
We now consider alternative approximations of the WFR flow obtained by combining the idea of SMC tempering with a proposal $M_n$ approximating the W flow. We consider two approximations of the W flow, ULA~\eqref{eq:ula} whose proposal kernel is given by~\eqref{eq:Rgamma} and the MALA \citep{roberts1996exponential}.

\paragraph{ULA proposal}
Consider $(\hat{\eta}_n)_{n\geq0} = (\mu_n)_{n\geq0} $ given by~\eqref{eq:tempering}, $M_n$ is the ULA kernel~\eqref{eq:Rgamma} and $L_{n-1}$ is the optimal backward kernel~\eqref{eq:optimal_L}. In this case, we obtain the weights~\eqref{eq:w_backward_opt}
\begin{align}
\label{eq:w_wfr_smc}
w_{n}(x_{n}) 
&= \frac{\mu_{n}(x_{n})}{\int \mathcal{N}(x_{n}; x_{n-1}+\gamma\nabla\log\pi(x_{n-1}), 2\gamma\cdot\textsf{Id})\mu_{n-1}(x_{n-1})dx_{n-1} }\\
&=\frac{\mu_{n-1}^{e^{-\gamma}}(x_n)\pi^{1-e^{-\gamma}}(x_n)}{\int \mathcal{N}(x_{n}; x_{n-1}+\gamma\nabla\log\pi(x_{n-1}), 2\gamma\cdot\textsf{Id})\mu_{n-1}(x_{n-1})dx_{n-1}}.\notag
\end{align}

Consider now two possible approximations for~\eqref{eq:w_wfr_smc}. 
For small $\gamma$, we have $\mu_{n-1}R_\gamma (x_{n}) \approx \mu_{n-1}(x_n)$ as the Gaussian approaches a Dirac delta centred at $x_{n-1}$. 
The first approximation is obtained by replacing $\mu_{n-1}$ in the numerator of~\eqref{eq:w_wfr_smc} with $\mu_{n-1}R_\gamma $.
This coincides with~\eqref{eq:weights_ideal} and the resulting algorithm is SMC-WFR in Algorithm~\ref{alg:smc_wfr}.
The main drawback of SMC-WFR is its computational cost: each weight~\eqref{eq:w_wfr} requires $\mathcal{O}(N)$ cost against the $\mathcal{O}(1)$ cost of SMC tempering.
For small number of particles $N$ this additional cost is balanced by the improved convergence speed (Section~\ref{sec:gm1}) but for large $N$ it might become prohibitive. 

The second approximation is obtained by replacing $\mu_{n-1}R_\gamma $ in~\eqref{eq:w_wfr_smc} with $\mu_{n-1}$. The resulting weights are
\begin{align}
\label{eq:w_smcula}
    w_{n}(x_n) &\approx \left(\frac{\pi(x_n)}{\mu_{n-1}(x_{n})}\right)^{1-e^{-\gamma}} = \left(\frac{\pi(x_n)}{\mu_{0}(x_{n})}\right)^{(1-e^{-\gamma})e^{-(n-1)\gamma}}
\end{align}
where we used the recursive structure of~\eqref{eq:tempering} to obtain the final expression.
Each weight can be computed in $\mathcal{O}(1)$ cost w.r.t. $N$. We call the resulting algorithm SMC-ULA.


\paragraph{MALA proposal}

Consider the case in which $M_n$ is given by one step of MALA with target $\pi$ and the sequence $(\hat{\eta}_n)_{n\geq0}$ is still given by~\eqref{eq:tempering}.
The main difference between the ULA proposal and the MALA one is that the latter leaves $\pi$ invariant and thus we can select the backward kernel
\begin{align}
\label{eq:backward_smcmala}
L_{n-1}(x_n, x_{n-1})  = \frac{M_n(x_{n-1}, x_n)\pi(x_{n-1})}{\int M_n(z, x_n)\pi(z)d z}=\frac{M_n(x_{n-1}, x_n)\pi(x_{n-1})}{\pi(x_n)},
\end{align}
which leads to the weights 
\begin{align*}
w_n(x_{n-1}, x_n) &=\frac{\hat{\eta}_n(x_n)\pi(x_{n-1})}{\hat{\eta}_{n-1}(x_{n-1})\pi(x_n)} = \frac{\mu_n(x_n)\pi(x_{n-1})}{\mu_{n-1}(x_{n-1})\pi(x_n)} = \left(\frac{\mu_{n-1}(x_n)}{\pi(x_n)}\right)^{e^{-\gamma}}\frac{\pi(x_{n-1})}{\mu_{n-1}(x_{n-1})}.
\end{align*}
Using again the recursive structure of~\eqref{eq:tempering} the above simplifies to
\begin{align*}
w_n(x_{n-1}, x_n) &=\frac{\hat{\eta}_n(x_n)\pi(x_{n-1})}{\hat{\eta}_{n-1}(x_{n-1})\pi(x_n)} = \left(\frac{\mu_{0}(x_n)}{\pi(x_n)}\right)^{e^{-n\gamma}}\left(\frac{\pi(x_{n-1})}{\mu_{0}(x_{n-1})}\right)^{e^{-(n-1)\gamma}}.
\end{align*}
The use of a MALA kernel has two advantages over the one described in the previous section: (a) each weight can be computed in $\mathcal{O}(1)$ cost w.r.t. $N$ (against the  $\mathcal{O}(N)$ cost of~\eqref{eq:w_wfr}) and (b) the discretisation bias introduced by replacing the continuous time Langevin diffusion with its time-discretised approximation~\eqref{eq:ula} is removed. We call Algorithm~\ref{alg:smc} with $(\hat{\eta}_n)_{n\geq0}$ given by~\eqref{eq:tempering}, proposal $M_n$ given by MALA targeting $\pi$ and backward kernel~\eqref{eq:backward_smcmala} SMC-MALA.

\paragraph*{Numerical Comparison}

\begin{figure}
    \centering
    \includegraphics[width=0.32\linewidth]{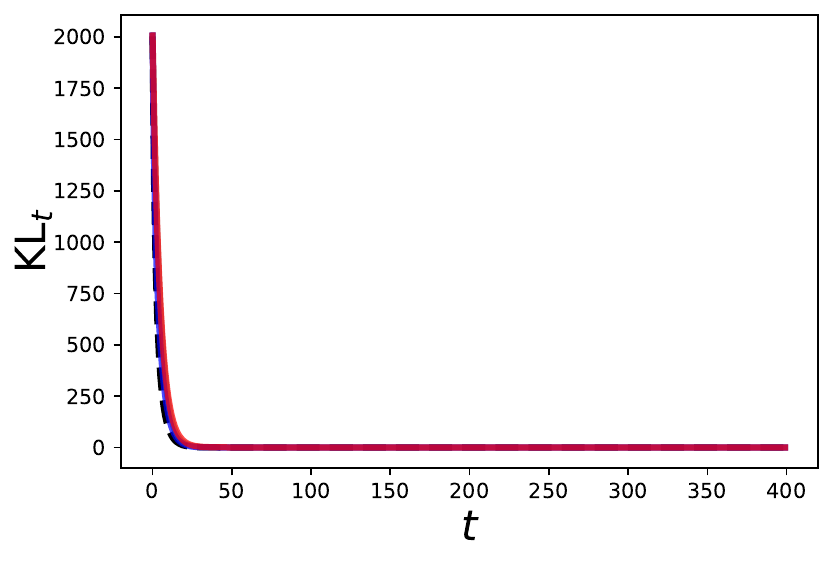}
    \includegraphics[width=0.32\linewidth]{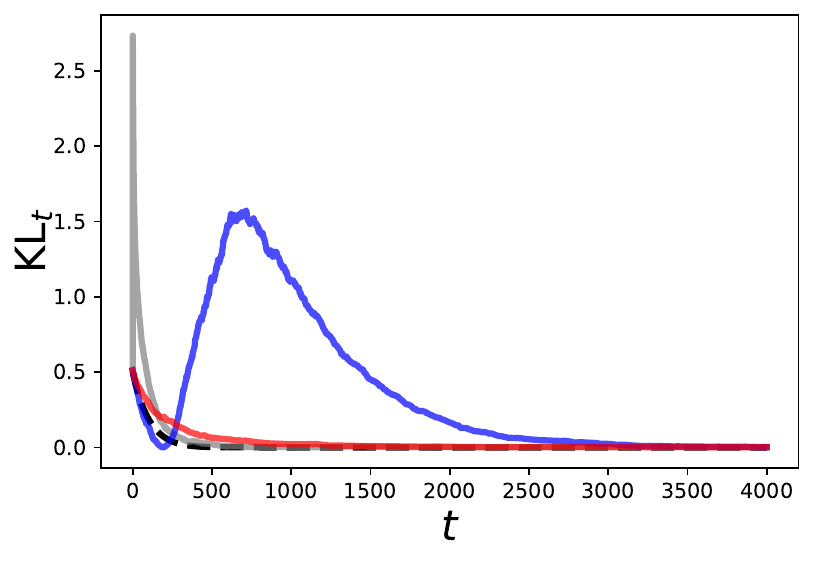}
    \includegraphics[width=0.32\linewidth]{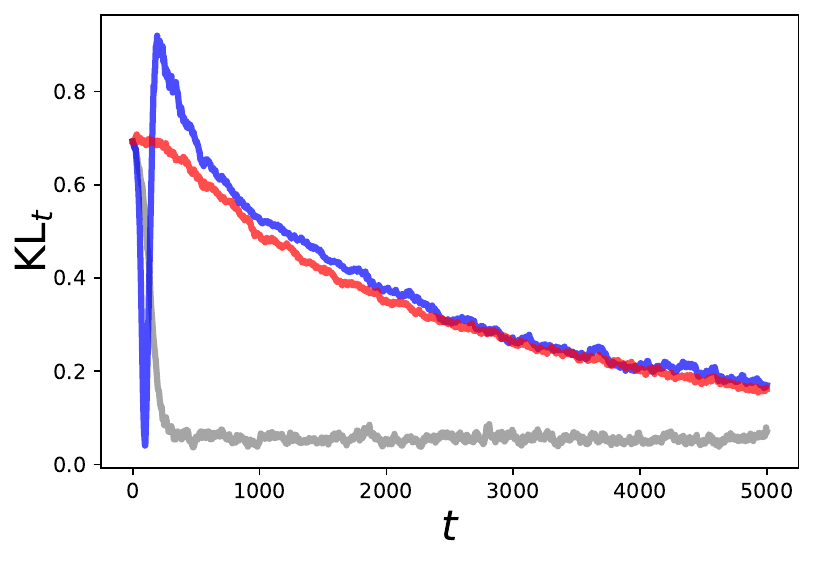}
    \includegraphics[width = 0.6\textwidth]{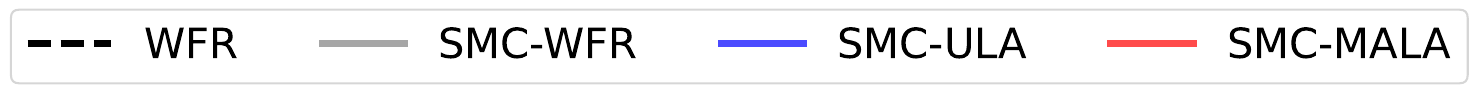}

    \caption{Comparison of evolution of $\KL$ of the exact WFR flow and approximations provided by SMC-WFR, SMC-ULA, SMC-MALA with initial distribution $\mu_0(x) = \mathcal{N}(x; 0, 1)$ and $\pi(x) = \mathcal{N}(x; 20, 0.1)$ (left), $\pi(x) = \mathcal{N}(x; 1, 5)$ (middle), $\pi(x) = \frac{1}{2}\mathcal{N}(x;0, 1) + \frac{1}{2}\mathcal{N}(x;6, 1)$ (right).
    }
    \label{fig:gaussians_exact_smc_variants}
\end{figure}

To compare SMC-WFR, SMC-ULA and SMC-MALA as approximations of WFR we consider three 1D targets
\begin{align*}
    \pi(x) &= \mathcal{N}(x; 20, 0.1), \qquad
    \pi(x) = \mathcal{N}(x; 1, 5), \qquad
    \pi(x) = \frac{1}{2}\mathcal{N}(x;0, 1) + \frac{1}{2}\mathcal{N}(x;6, 1).
\end{align*}
The initial distribution is a standard Gaussian. For the first two targets the exact evolution of KL along the WFR flow is available \citep{us_splitting} while for the last target we use numerical integration to compute KL along the iterations of SMC-WFR, SMC-ULA and SMC-MALA. All parameters (number of particles, number of iterations and time step) are equal.
Figure~\ref{fig:gaussians_exact_smc_variants} shows the evolution of $\KL$ given by SMC-WFR, SMC-ULA and SMC-MALA.
For the first target, convergence is driven by the W flow (approximated by ULA or MALA) and thus the three algorithms are all in agreement with the exact flow. 
For the second target, convergence is driven by the FR flow. Figure~\ref{fig:gaussians_exact_smc_variants} middle shows that SMC-ULA performs poorly in this case.
The third target is non-Gaussian, in this case both SMC-ULA and SMC-MALA exhibit slower convergence than SMC-WFR.
Thus, while related to SMC-WFR, SMC-ULA and SMC-MALA are poor approximations of the WFR flow and tend to be \textit{slower} than SMC-WFR.
This is confirmed by our experimental results in Section~\ref{sec:expe}: SMC-ULA and SMC-MALA are never faster than SMC-WFR.
Moreover, SMC-ULA can be unstable even if SMC-WFR is stable, we believe the use of the different approximations causes this behaviour and is linked to the fact that importance weights obtained from mixture proposals (as~\eqref{eq:w_wfr}) have lower variance than those obtained from single proposals (as~\eqref{eq:w_smcula}) \citep{elvira2019generalized}.

\section{Numerical Experiments}
\label{sec:expe}

We now present some numerical experiments to support our results in the previous sections. 
First, we show that SMC-WFR in Algorithm~\ref{alg:smc_wfr} outperforms the birth-death Langevin dynamics as approximation of WFR.
We then empirically confirm Proposition~\ref{prop:is_wfr} showing that SMC-WFR converges faster than with an approximation the W flow given by $N$ parallel ULA chains and the SMC approximation of the FR flow (i.e. tempering SMC) on targets with different characteristics.
Finally, we compare SMC-WFR with the variants introduced in Section~\ref{sec:smc} and other popoular sampling methods. Our aim is to identify scenarios in which the extra computational cost required by SMC-WFR is beneficial in achieving faster convergence to the target.
Due to the impossibility of accurately approximating $\KL$ in general scenarios, to evaluate convergence we consider 
the maximum mean discrepancy with standard Gaussian kernel (MMD; \cite{gretton2012kernel}) and the average Wasserstein-1 distance for each 1D marginal ($W_1$).
Code to reproduce our experiments is available at \small{ \url{https://github.com/FrancescaCrucinio/SMC-WFR}}.

\normalsize

\subsection{Birth-death Langevin and sequential Monte Carlo Wasserstein Fisher Rao}
\label{sec:bdl}

 \begin{figure}[t]
 	\centering
 	\begin{tikzpicture}[every node/.append style={font=\normalsize}]
  \node (img1) {\includegraphics[width = 0.3\textwidth]{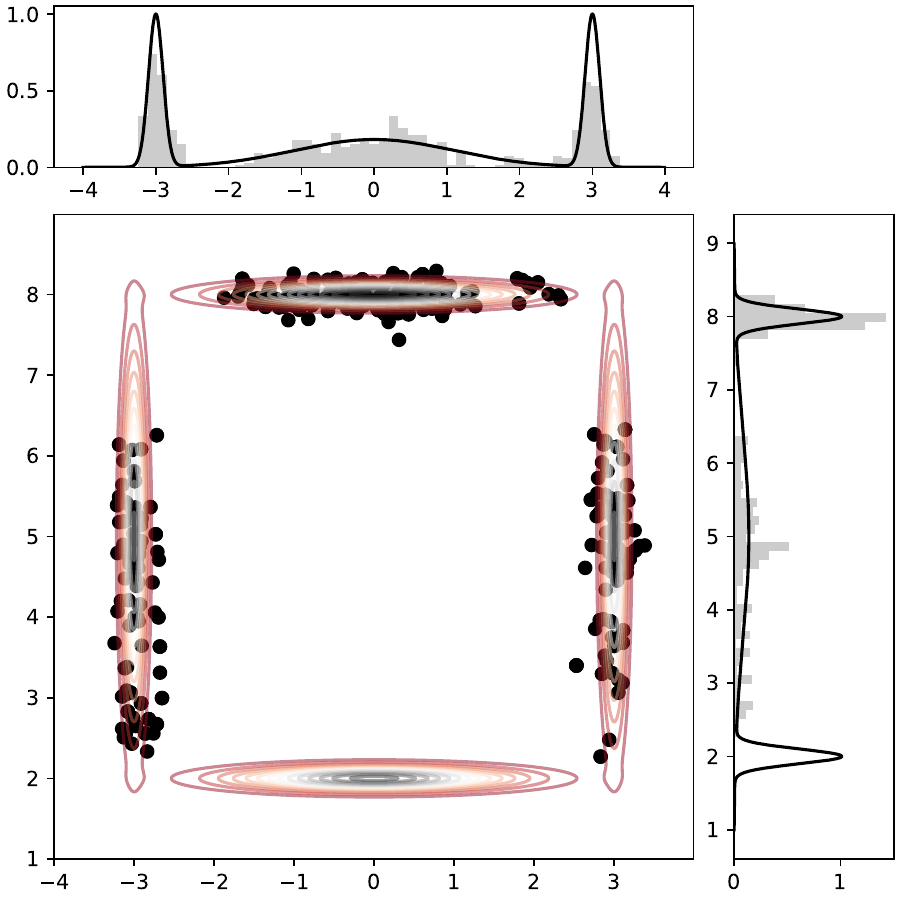}};
 		\node[right=of img1, node distance = 0, xshift = -1cm] (img2) {\includegraphics[width = 0.3\textwidth]{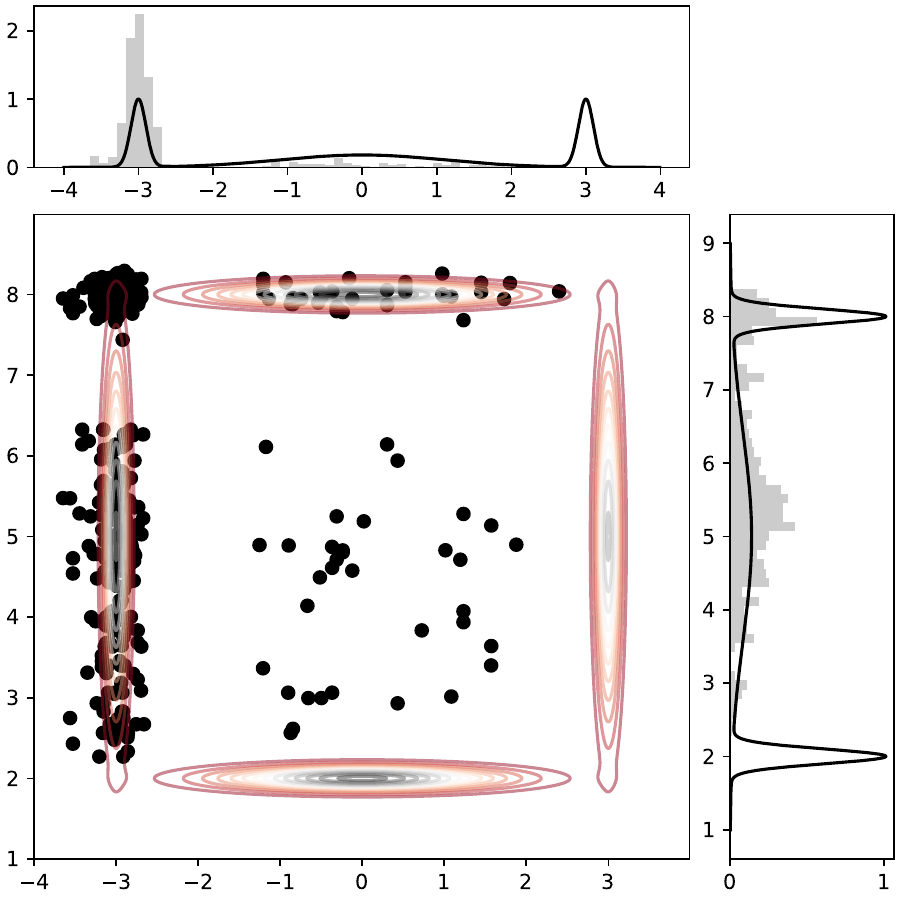}};
    \node[below=of img1, node distance = 0, anchor = center, yshift = 0.8cm] {BDL-PDE};
      \node[below=of img2, node distance = 0, anchor = center, yshift = 0.8cm] {BDL-KL};
   \node[right=of img2, node distance = 0, xshift = -1cm] (img3) {\includegraphics[width = 0.3\textwidth]{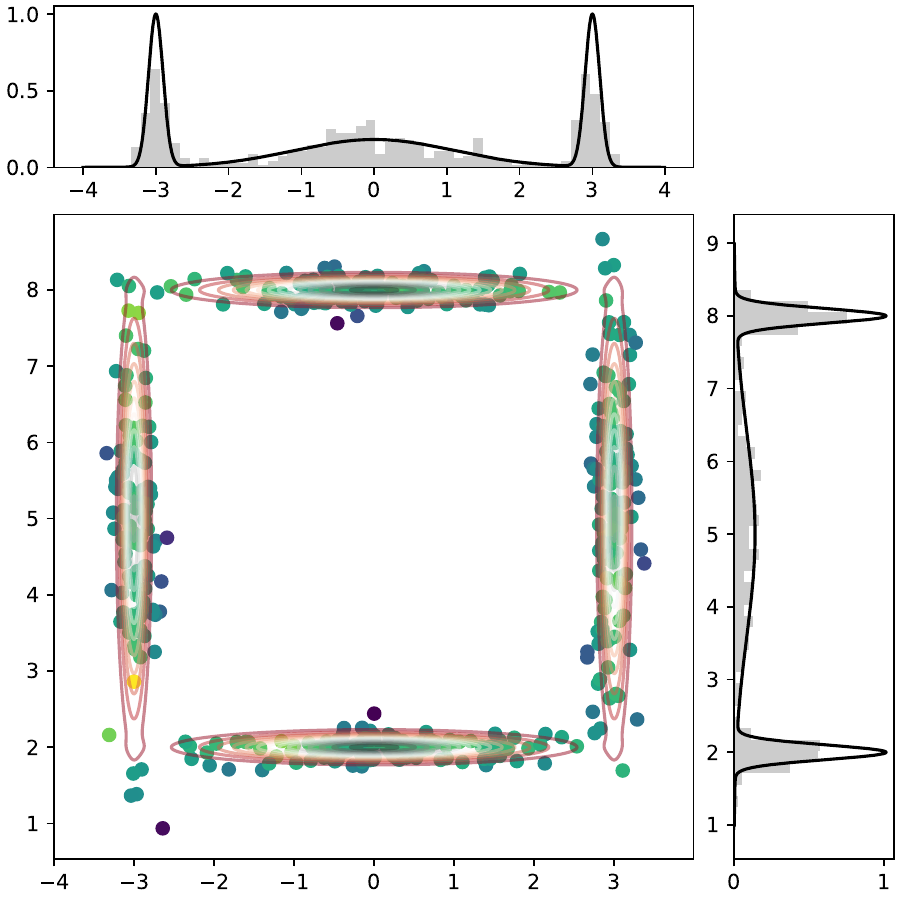}};
             \node[below=of img3, node distance = 0, yshift = 0.8cm] {SMC-WFR};
 	\end{tikzpicture}
 	\caption{Comparison of approximations of the target $\pi(x) = \sum_{i=1}^4 w_i \mathcal{N}(x; m_i, C_i)$ for the birth-death Langevin algorithms and our SMC approximation. For the latter the colour of the particles corresponds to the weight (brighter corresponds to higher weight). We compare both the joint distribution and the marginals.}
 	\label{fig:4modes}
 \end{figure}

We compare the approximation of the WFR flow~\eqref{eq:wfr} provided by  Algorithm~\ref{alg:smc_wfr} with that given by the birth-death Langevin (BDL) algorithms of \cite{lu2023birth, Lu2019}. 
We consider the two dimensional Gaussian mixture of \citet{Lu2019}, given by $\pi(x) = \sum_{i=1}^4 w_i \mathcal{N}(x; m_i, C_i)$, 
\begin{align*}
& w_i = 1/4, i =1,\cdots,4,\ m_1 = (0,8)^T,
m_2 = (0,2)^T,
m_3 = (-3,5)^T, m_4 = (3,5)^T,\\
& \Sigma_1 = \Sigma_2 = \begin{pmatrix}
              1.2 & 0\\
              0 & 0.01
             \end{pmatrix},
             \Sigma_3 = \Sigma_4 = \begin{pmatrix}
              0.01 & 0\\
              0&  2
             \end{pmatrix}.
\end{align*}

The initial distribution $\mu_0$ is a Gaussian with $m_0 = m_1$ and covariance matrix $\Sigma_0=0.3\cdot \textsf{Id}$. We use $N = 500$, $T=1000$, $\gamma = 0.01$ and select $h = \gamma$ for the kernel density estimator used within BDL.
We compare the results in terms of mean, variance, $W_1$ and MMD.
Both algorithms have $\mathcal{O}(N^2)$ cost due to the computation of the birth-death rates and weights, in fact, they also have the same wallclock time when using the same number of iterations and of particles (Table~\ref{tab:bdl_wfr}).

 \begin{table}[t]
 \centering
 \begin{small}
 \begin{tabular}{l|cc|cc|cc}
 Method & $\mse$ mean & $\mse$ covariance & $W_1$ & MMD & runtime (s) & MMD$<\epsilon$\\
 \hline\noalign{\smallskip}
 BDL-PDE \citep{Lu2019} & 1.841 & 4.579 & 1.275 & 0.116 & 2.93 & 986\\
 BDL-KL \citep{lu2023birth} & 2.326 & 5.505 & 1.440 & 0.146 & 2.17 & 989\\
 SMC-WFR & 0.006 & 0.044 & 0.133 & 0.007 & 2.68 & 323\\
 \end{tabular}
 \end{small}
 \caption{Comparison of the approximation of the WFR flow~\eqref{eq:wfr} provided by the birth-death Langevin (BDL) algorithms and Algorithm~\ref{alg:smc_wfr}. We compare the mean squared error $\mse$ for mean and covariance, the average Wasserstein-1 distance over dimension and the maximum mean discrepancy with Gaussian kernel (MMD). We also report the average runtime and the number of iterations needed to have MMD$<\epsilon$ for $\epsilon = 0.05$. All results are averaged over 50 replicates.}
 \label{tab:bdl_wfr}
 \end{table}

The results provided by our SMC approximation are considerably better than those obtained by the BDL algorithms. We attribute this improved performance to three factors. First, the approximations of the birth-death rate employed within the BDL algorithms are obtained using kernel density estimation (KDE) while the use of the convolution in~\eqref{eq:w_wfr} is motivated by the presence of the W flow and therefore does not add a layer of approximation as does KDE. Secondly, in SMC-WFR the reweighting step corresponds to the exact solution of the FR flow~\eqref{eq:fr_semigroup} in which $\mu_n$ is replaced by a particle approximation $\mu_n^N$; BDL does not use the exact solution but a discretisation of the FR flow obtained through birth-death dynamics. Thirdly, importance (re)sampling seems to be a more algorithmically stable way to approximate the FR flow in the sampling context as it ensures that the number of samples $N$ remains fixed with no additional cost while the BDL dynamics require additional steps to endure that $N$ is kept constant at each time step.

\subsection{Robustness to initial distribution}
\label{sec:gm1}

Figure~\ref{fig:gaussians_exact} suggests that WFR is more robust to the choice of initial distribution: when $\mu_0$ is more concentrated than $\pi$, WFR is closer to the FR flow, viceversa, WFR is closer to W when  $\mu_0$ is more diffuse than $\pi$.
We empirically show that this behaviour is maintained by Algorithm~\ref{alg:smc_wfr} when compared with 
ULA~\eqref{eq:ula} (as approximation of the W flow) and the SMC tempering algorithm in Section~\ref{sec:smc_fr} (as approximation of the FR flow).
We fix $\mu_0(x)=\mathcal{N}(x;0, 1)$ and consider two 1D Gaussian mixture targets
\begin{align*}
\pi_1(x) =& \frac{1}{2}\mathcal{N}(x;0, 1) + \frac{1}{2}\mathcal{N}(x;m, 1),\qquad \pi_2(x) = \frac{1}{2}\mathcal{N}(x;5, 0.1) + \frac{1}{2}\mathcal{N}(x;6, 0.1).
\end{align*}
The first target is challenging for the W flow as its log-Sobolev inequality increases exponentially with $m$ \citep{schlichting2019poincare}, while the second target is challenging for FR due to the presence of two very peaked modes.
For each algorithm we set $N=200$ and $\gamma=0.05$.
For ULA we run $N$ parallel chains using the same $\gamma$.  
Figure~\ref{fig:robust} shows the decay of $W_1$ against iteration number averaged over 50 repetitions for the three algorithms. 
For $\pi_1$ SMC tempering is faster than ULA and viceversa for $\pi_2$.
SMC-WFR is able to follow the behaviour of the fastest converging algorithm in both cases.
However, the cost of SMC-WFR is higher than ULA and SMC tempering.
When taking runtime into account, SMC-WFR achieves faster convergence than ULA and is comparable to SMC tempering for $\pi_1$, for $\pi_2$ SMC-WFR is faster than SMC tempering and comparable to ULA.

\begin{figure}
    \centering
    \includegraphics[width=0.42\linewidth]{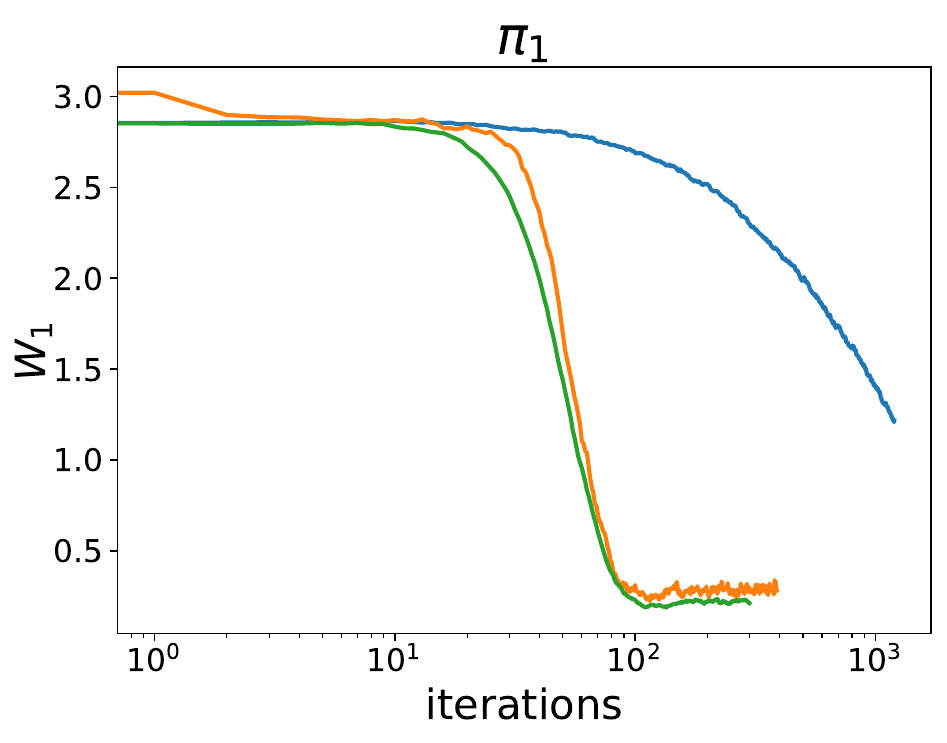}
    \includegraphics[width=0.4\linewidth]{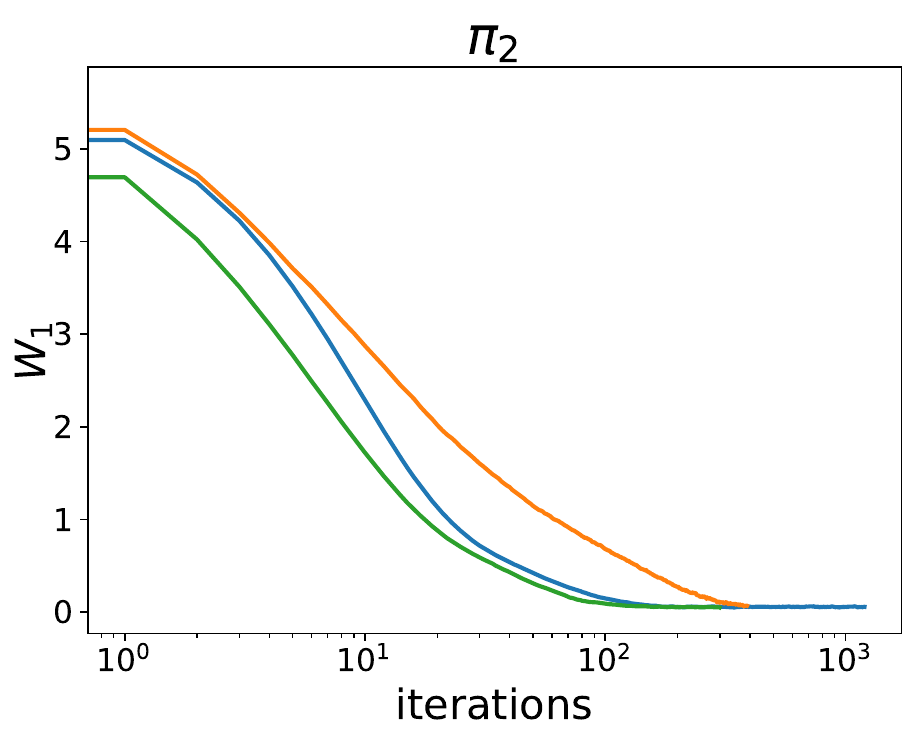}
    \includegraphics[width = 0.5\textwidth]{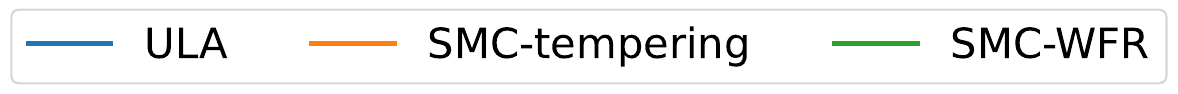}
    \caption{1D Gaussian mixture (Section~\ref{sec:gm1}): Evolution of $W_1$ for ULA, SMC tempering and SMC-WFR along iterations averaged over 50 repetitions. The initial distribution is $\mu_0(x) = \mathcal{N}(x; 0, 1)$. Left: $\pi_1$ with $m=6$ is more diffuse than $\mu_0$. Right: $\pi_2$ is more concentrated than $\mu_0$. }
    \label{fig:robust}
\end{figure}

\subsection{Comparison with other Monte Carlo algorithms}
\label{sec:comparison}
To assess whether the approximation of WFR in Algorithm~\ref{alg:smc_wfr} preserves the superior convergence speed of the WFR flow we compare SMC-WFR with several sampling algorithms in the literature.
In particular we consider $N$ parallel ULA chains~\eqref{eq:ula} and $N$ parallel MALA chains \citep{roberts1996exponential} and the SMC tempering algorithm in Section~\ref{sec:smc_fr}. We also include SMC-ULA and SMC-MALA.
We consider the following targets (see Appendix~\ref{app:targets} for details): 2D and 20D Gaussian mixture, 2D banana distribution, 2D donut distribution, 20D Gaussian.
We fix $\mu_0(x)=\mathcal{N}(x;\bf{0}, \textsf{Id})$.
We tune the variance of the proposal for MALA and SMC-MALA to achieve $57\%$ acceptance rate and the variance of RWM for SMC tempering to achieve $23\%$ acceptance rate \citep{roberts2001optimal}. 
For SMC-WFR and SMC-ULA we set $\gamma$ depending on the target, we use this $\gamma$ also for the reweighting step in SMC-MALA and SMC tempering (see Appendix~\ref{app:targets} for the experimental setup and additional results).
As the algorithms have different computational cost, we compare them both in terms of number of iterations and in terms of wall-clock time. We summarise our main takeaways as follows:
\begin{enumerate}[label=(\alph*)]
    \item \textbf{Shape of the target:} When ULA works well, e.g. targets which satisfy a log-Sobolev inequality~\eqref{eq:lsi} with small $C_{\textrm{LSI}}$, SMC-WFR and SMC-ULA (both approximating the WFR flow) are competitive in terms of number of iterations to convergence. In this case, convergence tends to be driven by the W flow and SMC-WFR and SMC-ULA perform similarly (Figure~\ref{fig:donut_app}, \ref{fig:20DG1} and~\ref{fig:20DG2}).  On targets with challenging geometries (Figure~\ref{fig:gm2d}) SMC-WFR achieves faster convergence in terms of wallclock time even when ULA, MALA, SMC-ULA and SMC-MALA fail to converge. 
    \item \textbf{Computational cost:} SMC-WFR is generally more expensive than the other methods due to the structure of the weights~\eqref{eq:w_wfr}, despite this, SMC-WFR tends to converge faster in wallclock time on targets with challenging geometries (Figure~\ref{fig:gm2d} and~\ref{fig:gm20d}), is competitive on targets where ULA works well (Figure~\ref{fig:donut_app}, \ref{fig:20DG1} and~\ref{fig:20DG2}), although ULA and MALA tend to be faster on these targets due to their lower computational cost, and can stabilise the behaviour of ULA in cases where $\nabla\log \pi$ is only locally Lipschitz continuous (Figure~\ref{fig:banana_app}).
    \item \textbf{Robustness:} All algorithms are sensitive to step size specification. MALA, SMC-tempering and SMC-MALA are more robust than ULA, SMC-WFR and SMC-ULA but this is largely due to the availability of simple tuning strategies for the proposal variance. SMC-WFR is more robust to the choice of initial distribution, performing well when the initial distribution is more diffuse than the target and when is more concentrated (Figure~\ref{fig:gm2d} and~\ref{fig:gm20d}).
    \item \textbf{Time discretisation bias:} Due to the lack of the Metropolis--Hastings accept/reject step, ULA, SMC-ULA and SMC-WFR exhibit a time discretisation bias. 
\end{enumerate}

\begin{figure}[h]
    \centering
     \includegraphics[width=0.9\linewidth]{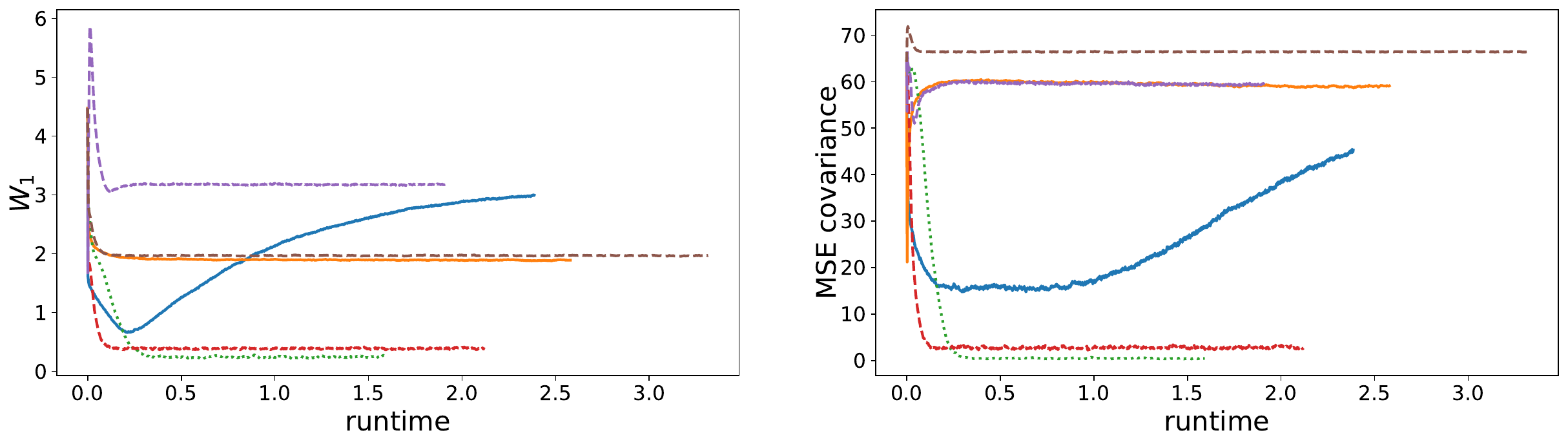}
     \includegraphics[width=0.9\linewidth]{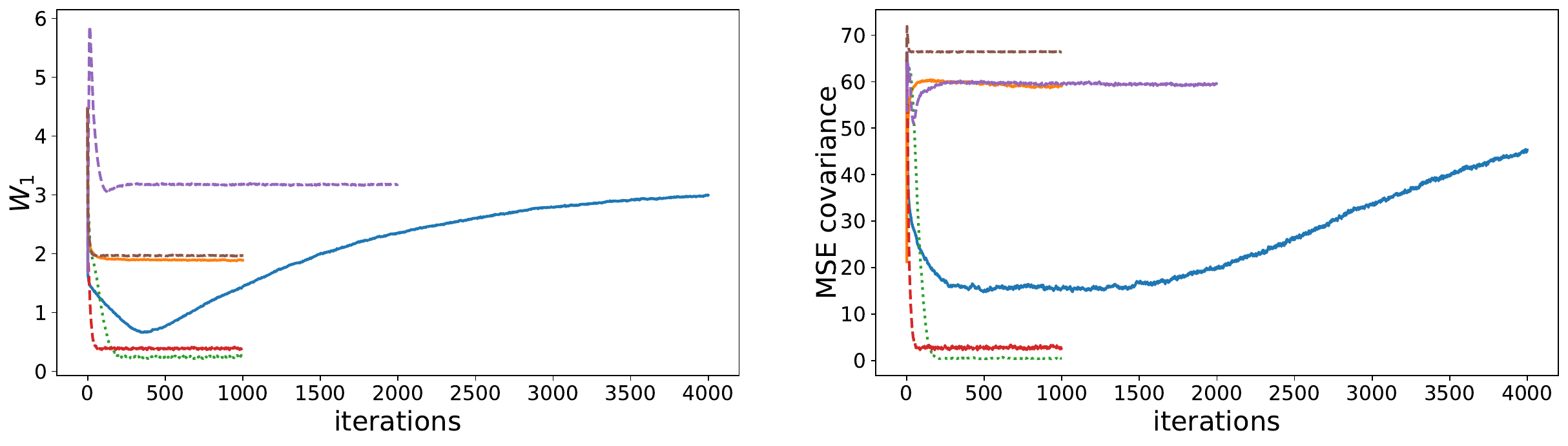}
    \includegraphics[width=0.45\linewidth]{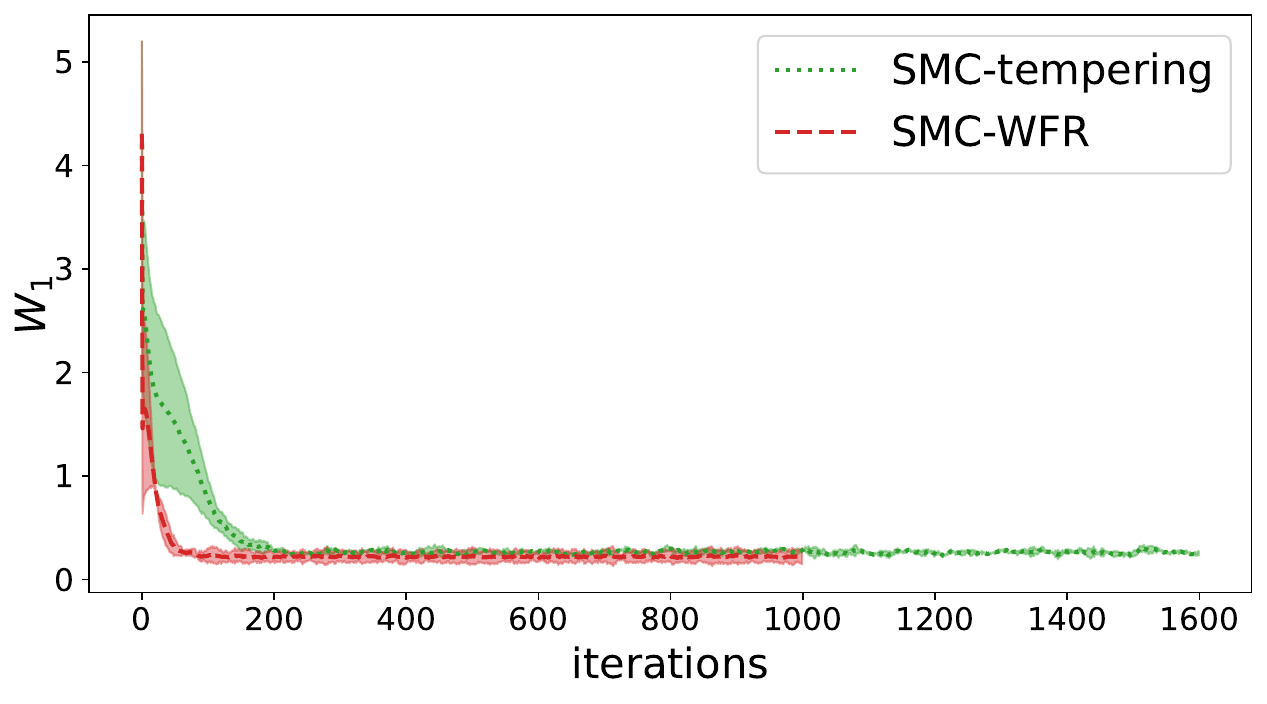}
     \includegraphics[width = 0.8\textwidth]{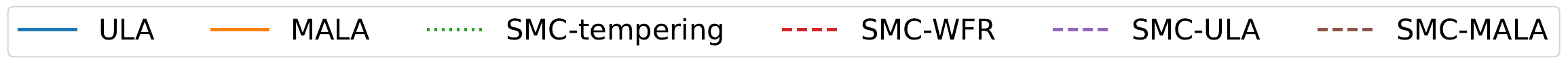}
    \caption{2D Gaussian mixture: Evolution of $W_1$ and mean squared error (MSE) for the covariance matrix against runtime (first row) and number of iterations (second row) averaged over 50 repetitions.  The initial distribution is $\mu_0(x) = \mathcal{N}(x; 0, \textsf{Id})$. We compare ULA, MALA, SMC tempering, SMC-WFR, SMC-ULA and SMC-MALA. SMC-WFR achieves the fastest convergence both in number of iterations and in runtime, ULA and MALA based algorithms fail to identify all the modes. Bottom plot: Evolution of $W_1$ for SMC-WFR and SMC tempering with 2$\times$standard deviation bands.}
    \label{fig:gm2d}
\end{figure}

\begin{figure}[h]
    \centering
     \includegraphics[width=\linewidth]{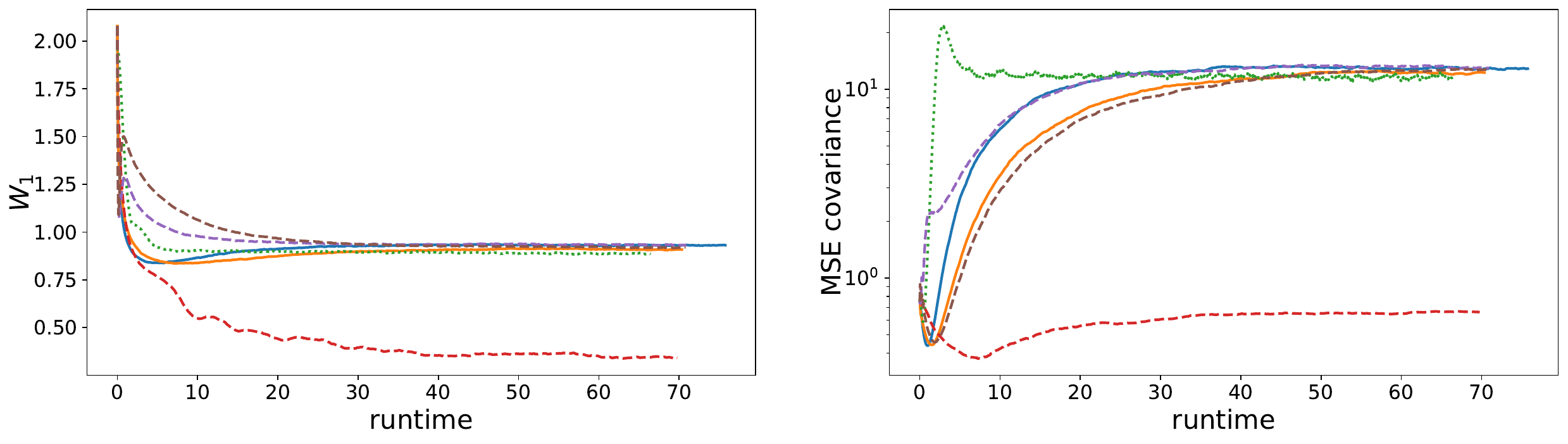}
     \includegraphics[width=\linewidth]{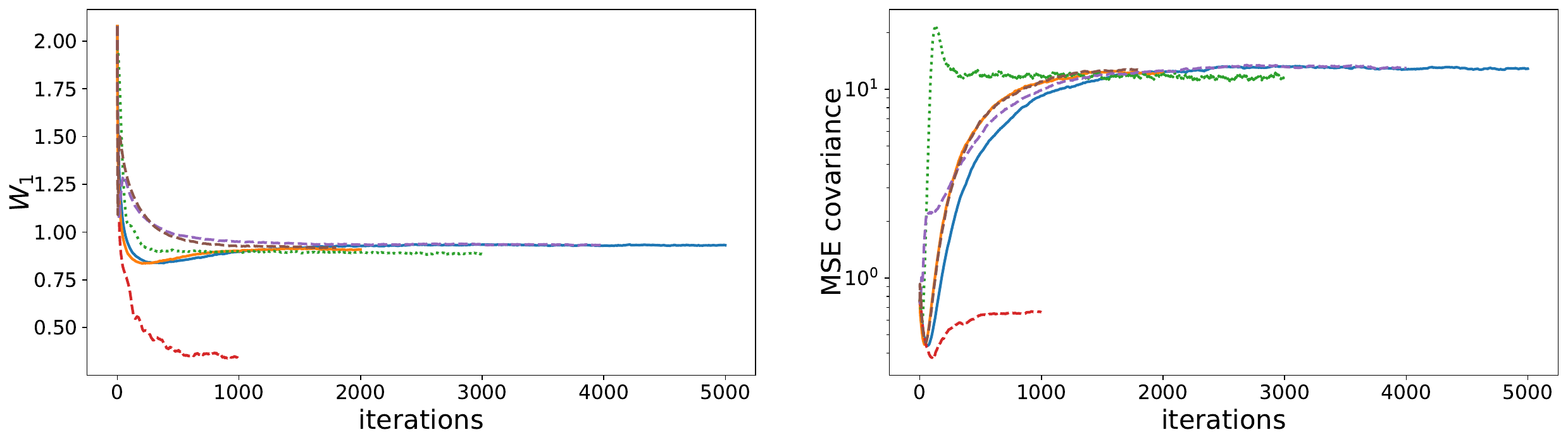}
     \includegraphics[width = 0.8\textwidth]{legend_comparison.pdf}
    \caption{20D Gaussian mixture: Evolution of $W_1$ and mean squared error (MSE) for the covariance matrix against runtime (first row) and number of iterations (second row) averaged over 50 repetitions.  The initial distribution is $\mu_0(x) = \mathcal{N}(x; 0, \textsf{Id})$. We compare ULA, MALA, SMC tempering, SMC-WFR, SMC-ULA and SMC-MALA. SMC-WFR achieves the fastest convergence both in number of iterations and in runtime.}
    \label{fig:gm20d}
\end{figure}


\section{Conclusions}
We study the connection between sampling algorithms and gradient flows in the space of probability measures, with focus on the Wasserstein--Fisher--Rao geometry. We identify a natural parallel between Fisher--Rao dynamics and importance sampling and use this idea to propose a sequential Monte Carlo (SMC) algorithm which approximates the Wasserstein--Fisher--Rao dynamics. This algorithm has convergence guarantees (Proposition~\ref{prop:is_wfr} and~\ref{prop:wlln}) and shows superior performance to other approximations (Section~\ref{sec:bdl}). 

Numerical approximations of the WFR flow are underexplored, and their ability to realise the theoretical gains of the WFR flow over pure W and FR flows is poorly understood. We propose a numerical approximation of the WFR flow (Algorithm~\ref{alg:smc_wfr}) and investigate whether this approximation maintains the superior theoretical performances of the WFR flow in Proposition~\ref{prop:is_wfr}. To achieve this goal, we performed a simulation study on a range of targets which exhibit different challenges (e.g. multimodality, concentration, high dimensionality) and compared the convergence speed in terms of wallclock time of Algorithm~\ref{alg:smc_wfr} with popular Monte Carlo methods. 

Interpreting SMC as the gradient flow of the KL divergence w.r.t. the Wasserstein--Fisher--Rao geometry allows us to exploit the well known results connecting Langevin dynamics to  Wasserstein gradient flows of KL to gain some insight on when SMC performs better than Langevin-based algorithms. In particular, Figure~\ref{fig:gaussians_exact_smc_variants} shows that naive combinations of ULA and importance sampling \emph{do not} lead to faster convergence.
We identify the following properties of the target and initial distribution which should guide the choice of sampling algorithm to achieve the best tradeoff between accuracy and convergence speed:
\begin{itemize}
    \item If $\pi\propto e^{-V_\pi}$ satisfies a log-Sobolev inequality~\eqref{eq:lsi} with small $C_{\textrm{LSI}}$, $\nabla V_\pi$ is well-behaved and $\mu_0$ puts mass on a region of high target probability, then the W flow works well, and ULA and MALA are likely to provide good results (e.g. the donut and Gaussian examples). In this case approximating the WFR flow might be unnecessary as it will closely follow the W flow.
    \item In the case of large $C_{\textrm{LSI}}$ or if the gradients of $\nabla V_\pi$ are not globally Lipschitz, the combination of the diffusive behaviour of ULA with the importance sampling approach approximating the FR step in SMC-WFR provides a considerable advantage both in terms of runtime and stability (e.g. the banana and Gaussian mixture examples). 
    \item SMC-WFR and SMC-tempering are more robust to the choice of initial distribution and provide good results even when $\pi$ and $\mu_0$ do not significantly overlap. This is largely due to the use of tempering techniques.
\end{itemize}
We also highlight some limitations of the proposed methodology and opportunities for further developments of numerical approximations of the WFR flow:
\begin{itemize}
    \item  Algorithm~\ref{alg:smc_wfr} relies on ULA for the proposal step and thus is not suited for non-differentiable targets, is unstable when ULA is and provides little exploration for targets with challenging geometry. To circumvent stability and non-differentiability issues, one could consider alternative approximations of the Langevin diffusion \citep{NEURIPS2019_6a8018b3_salim, salim2020wasserstein, mou2021high}. To improve exploration, standard recipes like preconditioning \citep{dalalyan2017theoretical, hird2025quantifying} and modulation of the diffusion coefficient could be beneficial \citep{Pavliotis2014}.
    \item The $\mathcal{O}(N^2)$ computational cost of SMC-WFR is balanced out by the improved convergence speed for small $N$ but for large $N$ the $\mathcal{O}(N)$ cost of the weights~\eqref{eq:w_wfr} might become prohibitive. This can be addressed adapting $N$-body learning techniques as in \cite{klaas2012toward, lang2005empirical} or considering efficient implementations using GPUs as shown in \citet[Section 4]{clarte2019collective} to obtain a $\mathcal{O}(N\log N)$ cost. The use of multiple proposal steps in SMC-WFR can also mitigate computational cost as the ULA updates are cheap, but is only beneficial when ULA works well. 
    \item SMC-WFR does not correct for the time discretisation bias of ULA. As swapping ULA for MALA as in SMC-MALA seems to deteriorate the performances of the algorithm, one could instead use adaptive stepsizes: large at first to favour fast convergence and then smaller as time goes on to reduce bias. Alternatively, one could replace SMC-WFR with SMC-MALA once SMC-WFR has converged to high target probability regions.
\end{itemize}

Nevertheless, we demonstrate that the proposed SMC based algorithm is competitive on a range of multi-modal examples.  We hope that the ideas in this manuscript help spur the development of further numerical approximations of WFR flows.

%% file: paper1_appendix.tex
\section{Proof of Proposition~\ref{prop:is_wfr}}
\label{app:proofs3}

Let us denote $\cF(\mu) =\KL(\mu||\pi)$.
To show that the mirror descent update~\eqref{eq:fr_semigroup} reduces $\cF$ we reproduce the argument of \citet[Proposition 1]{chopin2023connection}. Since $\cF$ is 1-relatively smooth with respect to $\phi: \mu\mapsto \int \log (\mu(x))d\mu(x)$ \citep{aubin2022mirror} and since $ 1-e^{-\gamma}\leq 1$, for all $\gamma>0$, we have 
\begin{align}
\label{eq:conv_smooth_new}
		\cF(\mu_{n})&\le \cF(\mu_{n-1/2})+ \ssp{\nabla \cF(\mu_{n-1/2}),\mu_{n}	-\mu_{n-1/2}}+B_{\phi}(\mu_{n}|\mu_{n-1/2})\\
  &\le \cF(\mu_{n-1/2}) + \ssp{\nabla \cF(\mu_{n-1/2}),\mu_{n}	-\mu_{n-1/2}}+\frac{1}{1-e^{-\gamma}} B_{\phi}(\mu_{n}|\mu_{n-1/2}),
\nonumber
	\end{align}
where $B_{\phi}(\nu|\mu)=\phi(\nu)-\phi(\mu)-\ssp{\nabla \phi(\mu), \nu- \mu}$ denotes the Bregman divergence and $\nabla \phi$ the first variation.
	Applying \citet[Lemma 3]{aubin2022mirror} to the convex function $G(\nu)=(1-e^{-\gamma}) \ssp{\nabla \cF(\mu_{n-1/2}),\nu-\mu_{n-1/2}}$, with $\mu=\mu_{n-1/2}$ and $\bar{\nu}=\mu_{n}$ yields
	\begin{align*}
		&\ssp{\nabla \cF(\mu_{n-1/2}),\mu_{n}	-\mu_{n-1/2}} + \frac{1}{1-e^{-\gamma}} B_{\phi}(\mu_{n}|\mu_{n-1/2}) \le\\
		&\qquad\qquad \ssp{\nabla \cF(\mu_{n-1/2}),\nu-\mu_{n-1/2}} + \frac{1}{1-e^{-\gamma}} B_{\phi}(\nu| \mu_{n-1/2}) - \frac{1}{1-e^{-\gamma}} B_{\phi}(\nu| \mu_{n}).
	\end{align*}
    Since $\cF$ is $1$-strongly convex relative to $\phi$, we also have:
	\begin{equation*}
		\ssp{\nabla \cF(\mu_{n-1/2}),\nu-\mu_{n-1/2}}\le \cF(\nu)-\cF(\mu_{n-1/2}) - B_{\phi}(\nu|\mu_{n-1/2}) 
	\end{equation*}
	Then \eqref{eq:conv_smooth_new} implies
    \begin{equation*}
		\cF(\mu_{n})\le \cF(\nu) + \left(\frac{1}{1-e^{-\gamma}}-1\right) B_{\phi}(\nu| \mu_{n-1/2}) - \frac{1}{1-e^{-\gamma}}B_{\phi}(\nu|\mu_{n}),
	\end{equation*} 
    for any $\nu$.
    We now consider two choices of $\nu$: $\nu=\mu_{n-1/2}$ and $\nu=\rho_\gamma$, whose existence is guaranteed by \cite[Corollary 1]{vempala2019rapid}. Given that both bounds hold we conclude that $\cF(\mu_n)$ is upper bounded by the minimum at each $n$.
    \paragraph{Case $\nu=\mu_{n-1/2}$.} We have
     \begin{align*}
     \cF(\mu_{n})&\le \cF(\mu_{n-1/2})-\frac{1}{1-e^{-\gamma}} B_{\phi}(\mu_{n-1/2}| \mu_{n})=-\frac{1}{1-e^{-\gamma}}\log\int \mu_{n-1/2}^{e^{-\gamma}}\pi^{1-e^{-\gamma}}
 \end{align*}
since $B_{\phi}(\mu_{n-1/2}| \mu_{n})=\KL(\mu_{n-1/2}|| \mu_{n})$ and $\mu_n \propto \mu_{n-1/2}^{e^{-\gamma}}\pi^{1-e^{-\gamma}}$. Notice that the rhs is $D_{e^{-\gamma}}
\bigl(\mu_{n-1/2}\|\pi\bigr)$ where $D_\alpha$ denotes the Rényi divergence of order $\alpha\in(0,1)$. Using Lemma~\ref{lemma:alpha_div} we have
\begin{align*}
  \cF(\mu_{n})&\le  \cF(\mu_{n-1/2})
-
\frac{1}{1-e^{-\gamma}}
\int_{e^{-\gamma}}^1
(s-e^{-\gamma})
\operatorname{Var}_{\nu_s}
\!\left(
\log\frac{\mu_{n-1/2}}{\pi}
\right)
ds\leq\cF(\mu_{n-1/2})
\end{align*}
where $\nu_s\propto \pi^{1-s}\mu_{n-1/2}^{s}$. Recall that  $\mu_{n-1/2}(x)  = \mu_{n-1}R_{\gamma}(x)$
and using \citet[Lemma 3]{vempala2019rapid} with $\gamma \leq C_\textrm{LSI}^{-1}/(4L_\pi^2)$
\begin{align*}
   \cF(\mu_{n})&\leq e^{-\gamma /C_\textrm{LSI}}\cF(\mu_{n-1})+6\gamma^2 d L_\pi^2-
\frac{1}{1-e^{-\gamma}}
\int_{e^{-\gamma}}^1
(s-e^{-\gamma})
\operatorname{Var}_{\nu_s}
\!\left(
\log\frac{\mu_{n-1/2}}{\pi}
\right)
ds.
\end{align*}
Proceeding by induction we then find that
\begin{align*}
\cF(\mu_n)
&\leq
e^{-n\gamma/C_\textrm{LSI}}\cF(\mu_0)
+
6\gamma^2 d L_\pi^2
\frac{1-e^{-n\gamma/C_\textrm{LSI}}}
{1-e^{-\gamma/C_\textrm{LSI}}}
\\
&\quad
-\frac{1}{1-e^{-\gamma}}
\sum_{k=1}^n
e^{-(n-k)\gamma/C_\textrm{LSI}}
\int_{e^{-\gamma}}^1
(s-e^{-\gamma})
\operatorname{Var}_{\nu_s}
\!\left(
\log\frac{\mu_{k-1/2}}{\pi}
\right)
\,ds.
\end{align*}

\paragraph{Case $\nu=\rho_\gamma$.}
We have
\begin{align*}
	\cF(\mu_{n})&\le \cF(\rho_\gamma) + \left(\frac{1}{1-e^{-\gamma}}-1\right) B_{\phi}(\rho_\gamma| \mu_{n-1/2}) - \frac{1}{1-e^{-\gamma}}B_{\phi}(\rho_\gamma|\mu_{n})\\
    &=\cF(\rho_\gamma )+ \left(\frac{1}{1-e^{-\gamma}}-1\right) B_{\phi}(\rho_\gamma R_\gamma| \mu_{n-1}R_\gamma) - \frac{1}{1-e^{-\gamma}}B_{\phi}(\rho_\gamma|\mu_{n})\\
    &\le \cF(\rho_\gamma )+ \left(\frac{1}{1-e^{-\gamma}}-1\right) B_{\phi}(\rho_\gamma | \mu_{n-1}) - \frac{1}{1-e^{-\gamma}}B_{\phi}(\rho_\gamma|\mu_{n})
\end{align*} 
where we used 
$\mu_{n-1/2}=\mu_{n-1}R_\gamma$ and the data processing inequality applied to $B_\phi = \KL$.
Multiplying by $e^{k\gamma}$ we obtain
\begin{align*}
e^{k\gamma}\cF(\mu_k)
&\leq
e^{k\gamma}\cF(\rho_\gamma)
+
\frac{e^{(k-1)\gamma}}{1-e^{-\gamma}}
B_{\phi}(\rho_\gamma\mid\mu_{k-1})
-
\frac{e^{k\gamma}}{1-e^{-\gamma}}
B_{\phi}(\rho_\gamma\mid\mu_k).
\end{align*}
Summing over $k=1,\ldots,n$, the Bregman terms telescope, yielding
\begin{align*}
e^{n\gamma}\cF(\mu_n)
\leq\sum_{k=1}^n e^{k\gamma}\cF(\mu_k)
&\leq
\sum_{k=1}^n e^{k\gamma}\cF(\rho_\gamma)
+
\frac{1}{1-e^{-\gamma}}
B_{\phi}(\rho_\gamma\mid\mu_0)
-
\frac{e^{n\gamma}}{1-e^{-\gamma}}
B_{\phi}(\rho_\gamma\mid\mu_n).
\end{align*}
We conclude that
\begin{align*}
\cF(\mu_n)
&\leq
\frac{1-e^{-n\gamma}}{1-e^{-\gamma}}
\KL(\rho_\gamma||\pi)
+
\frac{e^{-n\gamma}}{1-e^{-\gamma}}
B_{\phi}(\rho_\gamma\mid\mu_0).
\end{align*}

\begin{lemma}
\label{lemma:alpha_div}
Let \(\mu\) and \(\pi\) be probability measures such that \(\mu \ll \pi\). For
\(0<\alpha<1\), define
\[
D_\alpha(\mu\|\pi)
:=
\frac{1}{\alpha-1}
\log
\int
\left(\frac{d\mu}{d\pi}\right)^\alpha
d\pi.
\]
Further, for \(s\in[\alpha,1]\), define the probability measure
\(
\mu_s(dx)
\propto:
\left(\frac{d\mu}{d\pi}(x)\right)^s\pi(dx).
\)
Then, under mild integrability conditions,
\[
D_\alpha(\mu\|\pi)
=
\KL(\mu||\pi)
-
\frac{1}{1-\alpha}
\int_\alpha^1
(s-\alpha)
\operatorname{Var}_{\mu_s}
\!\left(
\log\frac{d\mu}{d\pi}
\right)
ds.
\]
\end{lemma}

\begin{proof}
Set
\(
\psi(s)
:=
\log
\int
\left(\frac{d\mu}{d\pi}\right)^s d\pi.
\)
Then \(\psi(1)=0\), while differentiation under the integral sign gives
\[
\psi'(s)
=
\int \log\frac{d\mu}{d\pi}\,d\mu_s,
\qquad
\psi''(s)
=
\operatorname{Var}_{\mu_s}
\!\left(\log\frac{d\mu}{d\pi}\right).
\]
Hence \(\psi'(1)=\KL(\mu\|\pi)\). Taylor's theorem with integral remainder gives
\[
\psi(\alpha)
=
(\alpha-1)\KL(\mu\|\pi)
+
\int_\alpha^1
(s-\alpha)\psi''(s)\,ds.
\]
Since \(D_\alpha(\mu\|\pi)=\psi(\alpha)/(\alpha-1)\), substituting the expression for \(\psi''\) gives the claimed identity. The inequality follows since the variance is nonnegative.
\end{proof}

\section{Proofs of Section~\ref{sec:proxy}}
\label{app:lp}

We recall that the initial distribution is $\mu_0$, $\mu_{n-1/2}$ denotes the distribution after application of the W flow over time step $\gamma$ and $\mu_n$ the distribution obtained by applying the FR flow over time step $\gamma$ to $\mu_{n-1/2}$.
The particle approximations obtained with Algorithm~\ref{alg:smc_wfr} are as follows: the particle approximation to $\predictive$ is denoted by $\predictiveN = N^{-1} \sum_{i=1}^N \delta_{X_n^i}$, $\widetilde{\mu}_n^N = \sum_{i=1}^N W_n^i \delta_{X_n^i}$ and $\mu_n^N = N^{-1} \sum_{i=1}^N \delta_{\widetilde{X}_n^i}$ both approximate $\update$ and denote the empirical distribution after importance weighting and after resampling, respectively.
For any distribution $\eta$ and any $\testfn\in\bounded$ we denote $\eta(\testfn):= \int \testfn(x)\eta( x)dx$, similarly for all empirical distributions $\eta^N:=N^{-1}\sum_{i=1}^N \delta_{X^i}$ we denote the corresponding average by $\eta^N(\testfn):= N^{-1}\sum_{i=1}^N \testfn(X^i)$.

We control each step of Algorithm~\ref{alg:smc_wfr} separately. In particular, as the W flow corresponds to an application of a Markov kernel and for the resampling step we consider the classical multinomial resampling framework, classical results in the SMC literature (e.g. \cite{douc2008limit}) control these steps. Our main result is the control of the reweighting step approximating the FR flow in Proposition~\ref{wlln:fr}.

\begin{proposition}[Control of approximate weights]
    \label{wlln:fr}
Under the assumptions of Proposition~\ref{prop:wlln}, assume that $\mu_{n-1}^N(\testfn)\pconverges \mu_{n-1}(\testfn)$
then $\predictiveN(\weightN\testfn) \pconverges\predictive(\weight\testfn)$, for any $\testfn\in \bounded$.
\end{proposition}
\begin{proof}
Let $\mathcal{F}_{n-1}^N$ denote the $\sigma$-field generated by the weighted samples up to (and including) time $n-1$.
Consider the decomposition
\begin{align}
\label{eq:decomposition_1}
\predictiveN(\weightN\testfn) - 
\predictive(\weight\testfn) & =\predictiveN(\weightN\testfn) - \E\left[\predictiveN(\weightN\testfn)\mid\mathcal{F}_{n-1}^N\right]\\
&+\E\left[\predictiveN(\weightN\testfn)|\mathcal{F}_{n-1}^N\right] -\predictive(\weight\testfn). \label{eq:decomposition_2}
\end{align}   

Let us start with~\eqref{eq:decomposition_1}. By the multinomial resampling step, conditional on
$\mathcal{F}_{n-1}^N$, the particles
are independent and identically distributed
with density $q_{n}^N$. Moreover, $\weightN$ is
$\mathcal{F}_{n-1}^N$-measurable. Therefore,
\begin{align*}
\Var\left(
\left.
\predictiveN(\weightN\testfn)
\right|
\mathcal{F}_{n-1}^N
\right)
&\leq
\frac{1}{N^2}
\sum_{i=1}^N
\E\left[
\left.
\testfn(X_n^i)^2
\weightN(X_n^i)^2
\right|
\mathcal{F}_{n-1}^N
\right]
\\
&=
\frac{1}{N}
\int_{\mathbb R^d}
\testfn(x)^2
w_n(x)^{2}
q_{n}^N(dx)\\
&\leq \supnorm{\testfn}^2\frac{1}{N}
\int_{\mathbb R^d}
w_n(x)^{2}
q_{n}^N(dx)\to 0
\end{align*}
by assumption~\eqref{ass:B}. Then the conditional Chebyshev's inequality implies $\eqref{eq:decomposition_1}\pconverges0$.

Consider now~\eqref{eq:decomposition_2} and observe that
   $ \predictive( \weight\testfn) = \int \mu_{n-1/2}(x)^{e^{-\gamma}} \pi(x)^{1-e^{-\gamma}}\varphi(x)dx$.
By Lemma~\ref{prop:conditional_expe}, we have
\begin{align*}
\E\left[
\predictiveN(\weightN\testfn)
\mid \mathcal{F}_{n-1}^N
\right]
-
\predictive(\weight\testfn)
&=
\int_{\mathbb{R}^d}
\testfn(x)\pi(x)^{1-e^{-\gamma}}
\left[
\left(q_n^N(x)\right)^{e^{-\gamma}}
-
\left(\mu_{n-1/2}(x)\right)^{e^{-\gamma}}
\right]dx.
\end{align*}
For every $N$ and every $x$,
\(
q_n^N(x)
\le
(4\pi\gamma)^{-d/2},
\)
because a Gaussian density with covariance $2\gamma\cdot\textsf{Id}$ is bounded above by $(4\uppi\gamma)^{-d/2}$. Similarly, $\supnorm{\mu_{n-1/2}}\leq (4\uppi\gamma)^{-d/2}$. Thus
\begin{align*}
\left|\E\left[
\predictiveN(\weightN\testfn)
\mid \mathcal{F}_{n-1}^N
\right]
-
\predictive(\weight\testfn)\right|
&\leq 2\supnorm{\testfn}(4\pi\gamma)^{-d e^{-\gamma}/2}
\int_{\mathbb{R}^d}\pi(x)^{1-e^{-\gamma}}\,dx<\infty
\end{align*}
thanks to the integrability assumption.
Observe that $q_n^N(x) = \int_{\R^d}R_\gamma(y, x)\mu_{n-1}^N(dy)$ where $\supnorm{R_\gamma}\leq (4\uppi\gamma)^{-d/2}$. Thus the assumption guarantess that
$q_n^N(x)  \pconverges\int_{\R^d}R_\gamma(y, x)\mu_{n-1}(dy) =\mu_{n-1/2}(x)$. Then \citet[Theorem (BDD), Section 13.6]{williams1991probability} gives $\eqref{eq:decomposition_2}\pconverges0$.
Since both terms converge to 0 in probability, the result follows.
\end{proof}

Key to the previous result is the following Lemma which obtains the conditional expectations of $\predictiveN(\weightN\testfn)$ and closely follows the result in \citet[Proposition 7]{crucinio2023properties}.
\begin{lemma}
\label{prop:conditional_expe}
Let $\mathcal{F}_{n-1}^N$ denote the $\sigma$-field generated by the weighted samples up to (and including) time $n-1$. We have, for all $n\geq 2$ and all $\testfn\in\bounded$,
\begin{align*}
\E\left[\predictiveN(\weightN\testfn)\mid \mathcal{F}_{n-1}^N\right] 
&=\int \testfn(x)\left(q_{n}^N(x)\right)^{e^{-\gamma}}\pi(x)^{1-e^{-\gamma}}dx.
\end{align*}
\end{lemma}
\begin{proof}
Let us denote the $\sigma$-field generated by the particle system up to (and including) time $n-1$ before the W flow at time $n$ by $\mathcal{G}_{n-1}^N:= \sigma\left(\widetilde{X}_{n-1}^{i}: i\in\lbrace 1,\ldots,N\rbrace\right)\vee \mathcal{F}_{n-1}^N$.

Consider the conditional expectation
 \begin{align}
 \label{eq:conditional_expe}
 \E\left[\predictiveN(\weightN\testfn)\mid \mathcal{F}_{n-1}^N\right] &=\frac{1}{N}\sum_{i=1}^N\E\left[\weightN(X_n^i)\testfn(X_n^i)|\mathcal{F}_{n-1}^N\right]\\
&=\sum_{j=1}^N\frac{w_{n-1}(X_{n-1}^j)}{\sum_{k=1}^Nw_{n-1}(X_{n-1}^k)}\int R_\gamma(X_{n-1}^j, dx_n)\weightN(x_n)\testfn(x_n),\notag
\end{align}
where the third equality follows from the fact that  $X_n^i|\mathcal{G}_{n-1}^N \sim R_\gamma(\widetilde{X}_{n-1}^i, \cdot)$ for each $i=1,\dots, N$ and the fourth from the fact that  $\{X_{n-1}^j\}_{j=1}^N$ and $w_{n-1}(X_{n-1}^j)$ are $\mathcal{F}_{n-1}^N$-measurable and conditionally each $\widetilde{X}_{n-1}^i$ is drawn independently from the categorical distribution with probabilities given by the weights.
Plugging the definition of $\weightN$ into the above and using
\begin{align*}
    q_{n}^N(x)&=\sum_{i=1}^N W_{n-1}^i R_\gamma(X_{n-1}^i, x)=\sum_{i=1}^N\frac{w_{n-1}(X_{n-1}^i)}{\sum_{k=1}^Nw_{n-1}(X_{n-1}^k)}R_\gamma(X^i_{n-1}, x)
\end{align*}
(with a slight abuse of notation we denote the density of~\eqref{eq:Rgamma} by $R_\gamma$) we obtain the result.

\end{proof}

\subsection{Proof of Proposition~\ref{prop:wlln}}
The proof proceeds inductively.
At $n=0$ the particles $(X_0^{i})_{i=1}^N$ are sampled i.i.d. from $\mu_0$. For $n=1$, no importance resampling is performed and
$\widetilde \mu_0\equiv \mu_0$, so the standard weak law of large numbers (WLLN) guarantees that $\widetilde{\mu}_0^N(\testfn)=\mu_0^N(\testfn)\pconverges\mu_0(\testfn)$.
Suppose then that $\mu_{n-1}^N(\testfn)\pconverges \mu_{n-1}(\testfn)$ for all $\testfn\in\bounded$.
\citet[Theorem 3]{douc2008limit} then shows that $\widetilde{\mu}_{n-1}^N$ obtained by resampling satisfies $\widetilde{\mu}_{n-1}^N(\testfn)\pconverges \mu_{n-1}(\testfn)$.  As the W flow corresponds to application of the Markov kernel~\eqref{eq:Rgamma}
\citet[Theorem 1]{douc2008limit} with importance weights therein set to $1/N$ guarantees $\mu_{n-1/2}^N(\testfn)\pconverges \mu_{n-1/2}(\testfn)$.
Consider now 
\[
\mu_n^N(\testfn)=\sum_{i=1}^N
W_n^i\varphi(X_n^i)
=
\frac{
N^{-1}
\sum_{i=1}^N
\varphi(X_n^i)
w_n(X_n^i)
}{
N^{-1}
\sum_{i=1}^N
w_n(X_n^i)
}
=
\frac{\mu_{n-1/2}^N(w_n\testfn)}{\mu_{n-1/2}^N(w_n)}.
\]
Proposition~\ref{wlln:fr} and the continuous mapping theorem give
\begin{align*}
\mu_n^N(\testfn)=\frac{\mu_{n-1/2}^N(w_n\testfn)}{\mu_{n-1/2}^N(w_n)}&\pconverges\frac{\mu_{n-1/2}(w_n\testfn)}{\mu_{n-1/2}(w_n)}=\frac{\int \mu_{n-1/2}(x)\varphi(x)\left( \frac{\pi(x)}{\mu_{n-1/2}(x)}\right)^{1-e^{-\gamma}}dx}{\int \mu_{n-1/2}(x)\left( \frac{\pi(x)}{\mu_{n-1/2}(x)}\right)^{1-e^{-\gamma}}dx} =\mu_n(\varphi).
\end{align*}
This completes the induction.
Using standard techniques (e.g. \cite{berti2006almost}) given in detail for the context of interest in \citet[Supplementary Material, Theorem 1]{schmon2018large}, the WLLN established above can be strengthened to the convergence of the measures in the weak topology.

\subsection{Proof of Proposition~\ref{prop:w2}}

Lemma~\ref{lemma:second_moment} shows that  $ \mu_n \in \cP_2(\R^d)$ and $ \mu_n^N \in \cP_2(\R^d)$ almost surely. Proposition~\ref{prop:wlln} shows that $\mu_n^N$ converges weakly to $\mu_n$ in probability. To obtain $W_2(\mu_n^N, \mu_n)\pconverges 0$ it remains to show convergence of the second moment.
Fix $n$, $\varepsilon>0$ and $\delta>0$. For $R>0$, define
\(
\varphi_R(x):=\norm{x}^2\mathbf{1}_{\{|x|\leq R\}}.
\)
For every $N$, we have the decomposition
\[
\begin{aligned}
\left|
\int_{\mathbb{R}^d}\norm{x}^2\,\mu_n^N(dx)
-
\int_{\mathbb{R}^d}\norm{x}^2\,\mu_n(dx)
\right|
\leq{}&
\left|\mu_n^N(\varphi_R)
-
\mu_n(\varphi_R)
\right|
+
\int_{\{|x|>R\}} \norm{x}^2\,\mu_n^N(dx)
+
\int_{\{|x|>R\}} \norm{x}^2\,\mu_n(dx)\\
&:= T_1(R, N)+ T_2(R, N)+T_3(R).
\end{aligned}
\]
Lemma~\ref{lemma:second_moment} shows that $\mu_n$ has finite second moment for all $n\geq 0$. Thus, there exists $R_1<\infty$ such that for all $R\geq R_1$
\(
T_3(R)<\frac{\varepsilon}{3}.
\) 
Lemma~\ref{lemma:ui_p} with $\delta, \varepsilon$ replaced by $\delta/2, \varepsilon/3$ therein, shows that there exists $R_2<\infty,\ N_0<\infty$ such that, for all $N\geq N_0$, 
$\mathbb{P}\left(T_2(R_2, N)>\frac{\varepsilon}{3}\right)<\frac{\delta}{2}$.
Since the tail integral is non-increasing in $R$, for every $R\geq R_2$ and
$N\geq N_1$, we also have \(
\mathbb{P}\left(T_2(R, N)>\frac{\varepsilon}{3}\right)
<\frac{\delta}{2}.
\)
Set \(
R:=\max\{R_1,R_2\}
\). Since $\varphi_R$ is bounded, by Proposition~\ref{prop:wlln} there exists $N_1<\infty$ such that, for every $N\geq N_1$,
\(
\mathbb{P}\left(T_1(R, N)
>\frac{\varepsilon}{3}
\right)
<\frac{\delta}{2}.
\)

Let $N_2:=\max\{N_1,N_0\}$. Then for all $N\geq N_2$
\[
\begin{aligned}
&\mathbb{P}\left(
\left|
\int_{\mathbb{R}^d}\norm{x}^2\,\mu_n^N(dx)
-
\int_{\mathbb{R}^d}\norm{x}^2\,\mu_n(dx)
\right|
>\varepsilon
\right)
<
\frac{\delta}{2}+\frac{\delta}{2}+0
=\delta.
\end{aligned}
\]
Since $\varepsilon>0$ and $\delta>0$ were arbitrary, it follows  that
\(
\int_{\mathbb{R}^d}\norm{x}^2\,\mu_n^N(dx)
\pconverges
\int_{\mathbb{R}^d}\norm{x}^2\,\mu_n(dx).
\)

\subsubsection{Auxiliary Results}

\begin{lemma}[Second moment]
\label{lemma:second_moment}
Under the assumptions of Proposition~\ref{prop:w2}, $ \mu_n \in \cP_2(\R^d)$ and $ \mu_n^N \in \cP_2(\R^d)$ almost surely for all $n\geq 0$. 
\end{lemma}
\begin{proof}
    Since $V_\pi$ is $L_\pi$-smooth and strongly convex, $\pi$ has finite second moment \cite[Proposition 1(ii)]{10.3150/18-BEJ1073}. In addition, since $\mu_0$ has finite second moment the ULA kernel~\eqref{eq:Rgamma} gives $\mu_{1/2}$ with finite second moment \cite[Proposition 2(i)]{10.3150/18-BEJ1073} and the FR update~\eqref{eq:fr_semigroup} also preserves finiteness of the second moment \cite[Eq. (B.5)]{chen2023sampling}. An inductive argument then gives the result. 
    Consider now $\int_{\R^d}\norm{x}^2\mu_n^N(dx)$. For $n=0$ the result is trivially true since $X_0^i\sim\mu_0$. For general $n>0$ we have
    \begin{align*}
    \int_{\R^d}\norm{x}^2\mu_n^N(dx) = \sum_{i=1}^N W_n^i\norm{X_n^i}^2 \leq \sum_{i=1}^N \norm{X_n^i}^2 
    \end{align*}
    since $0\leq W_n^i\leq 1$. We now show that $\norm{X_n^i}^2$ is finite for each $i=1, \dots, N$. 
    Suppose that $\norm{X_{n-1}^i}^2$ is finite almost surely for every $i=1,\dots,N$. At the resampling step, each $\widetilde{X}_{n-1}^i$ is drawn from the discrete probability measure $\sum_{j=1}^N W_{n-1}^j\delta_{X_{n-1}^j}$
Consequently, $\widetilde{X}_{n-1}^i=X_{n-1}^{J_i}$ for some random index $J_i\in{1,\dots,N}$, and hence
\(
\norm{\widetilde{X}_{n-1}^i}^2
\leq
\max_{1\leq j\leq N}\norm{X_{n-1}^j}^2\)
which is finite almost surely. 
    Application of the ULA kernel~\eqref{eq:ula} preserves finiteness of the second moment \cite[Proposition 2(i)]{10.3150/18-BEJ1073} and thus $\norm{X_n^i}^2<\infty$ almost surely for every $i=1,\dots,N$.  Thus $\mu_n^N\in\mathcal P_2(\mathbb R^d)$ almost surely for every $n\geq0$.
   
\end{proof}

\begin{lemma}[Uniform integrability]
\label{lemma:ui_p}
Under the assumptions of Proposition~\ref{prop:w2}, for all $\varepsilon>0, \delta>0$ there exists $R<\infty,\ N_0<\infty$ such that, for all $N\geq N_0$,
\begin{equation}
\label{ass:UI_P}
\mathbb{P}\left(
\int_{\{|x|>R\}} \norm{x}^2\,\mu_n^N(dx)
>\varepsilon
\right)<\delta.
\end{equation}
\end{lemma}
\begin{proof}
We prove~\eqref{ass:UI_P} by induction on $n$.
For $n=0$, recall that $X_0^1,\ldots,X_0^N$ are i.i.d. with law $\mu_0$. 
Since $\mu_0\in\mathcal{P}_2(\mathbb{R}^d)$, for every
$\varepsilon,\delta>0$ there exists $R<\infty$ such that $\int_{\{\|x\|>R\}}
\|x\|^2\,\mu_0(dx)<\delta \varepsilon$.
Markov's inequality then gives
\begin{align*}
\mathbb{P}\left(
\int_{\{\|x\|>R\}}
\|x\|^2\,\mu_0^N(dx)
>\varepsilon
\right)
\leq 
\frac{1}{\varepsilon}\mathbb{E}\left[
\int_{\{\|x\|>R\}}
\|x\|^2\,\mu_0^N(dx)
\right]
&=
\frac{1}{\varepsilon}\int_{\{\|x\|>R\}}
\|x\|^2\,\mu_0(dx)<\delta
\end{align*}
for every $N$. Thus~\eqref{ass:UI_P} holds for $n=0$.

Now fix $n\geq1$ and suppose that~\eqref{ass:UI_P} holds for
$\mu_{n-1}^N$. 
First, since $\mu_{n-1}^N$ is a probability measure, we have
\[
\int_{\mathbb{R}^d}\|y\|^2\,\mu_{n-1}^N(dy)
=
\int_{\{\|y\|\leq R\}}\|y\|^2\,\mu_{n-1}^N(dy)
+
\int_{\{\|y\|>R\}}\|y\|^2\,\mu_{n-1}^N(dy)
\leq R^2+\int_{\{\|y\|>R\}}\|y\|^2\,\mu_{n-1}^N(dy).
\]
Fix $\rho>0$ and $\varepsilon>0$. By the induction hypothesis
\eqref{ass:UI_P} for $\mu_{n-1}^N$ applied with $\delta=\rho$, there exist $R<\infty$ and
$N_0$ such that, for all $N\geq N_0$,
\[
\mathbb{P}\left(
\int_{\{\|y\|>R\}}\|y\|^2\,\mu_{n-1}^N(dy)>\varepsilon
\right)<\rho.
\]
Set $M=R^2+\varepsilon$. Then, for all $N\geq N_0$,
\begin{align}
\label{eq:first_step}
\mathbb{P}\left(
\int_{\mathbb{R}^d}\|y\|^2\,\mu_{n-1}^N(dy)>M
\right)
&\leq
\mathbb{P}\left(
R^2+
\int_{\{\|y\|>R\}}\|y\|^2\,\mu_{n-1}^N(dy)>M
\right)\notag \\
&=
\mathbb{P}\left(
\int_{\{\|y\|>R\}}\|y\|^2\,\mu_{n-1}^N(dy)>\varepsilon
\right)<\rho.
\end{align}

We now prove that~\eqref{ass:UI_P} also holds for $\mu_n^N$.
Recall that
\[
\int_{\{|x|>R\}} \norm{x}^2\,\mu_n^N(dx)
= \frac{
\frac{1}{N}
\sum_{i=1}^N
\norm{X_n^i}^2
\mathbf{1}_{\{\norm{X_n^i}>R\}}
w_n(X_n^i)}{\frac{1}{N}
\sum_{i=1}^N w_n(X_n^i)}
:=\frac{A_{n,R}^N}{Z_n^N}.
\]
Proposition~\ref{prop:wlln} and Proposition~\ref{wlln:fr} guarantee that $Z_n^N \pconverges Z_n:= \int_{\R^d} \pi(x)^{1-e^{-\gamma}}\mu_{n-1/2}(x)^{e^{-\gamma}}dx>0$. Set \(c=Z_n/2>0\). Then there exists \(N_1\in\mathbb{N}\) such that, for every \(N\geq N_1\),
\begin{align*}
\mathbb{P}\left(Z_n^N<c\right)
\leq
\mathbb{P}\left(
\left|Z_n^N-Z_n\right|>\frac{Z_n}{2}
\right)
<
\frac{\delta}{2}.
\end{align*}
Then
\begin{align}
\label{eq:proba_Z}
\mathbb{P}\left(
\int_{\{\|x\|>R\}}
\|x\|^2\,\mu_n^N(dx)
>\varepsilon
\right)
\leq
\mathbb{P}\left(
A_{n,R}^N>c\varepsilon
\right)
+
\mathbb{P}\left(
Z_n^N<c
\right)=\mathbb{P}\left(
A_{n,R}^N>c\varepsilon
\right)
+\frac{\delta}{2}.
\end{align}
Consider $A_{n,R}^N$. Fix $\eta=c\varepsilon \delta/4$, then we have
\begin{align}
\label{eq:AN_decomposition}
\mathbb{P}\left(A_{n,R}^N>c\varepsilon\right)
&\leq
\mathbb{P}\left(
\mathbb{E}\left[
A_{n,R}^N
\mid\mathcal{F}_{n-1}^N
\right]>\eta
\right)+
\mathbb{P}\left(
A_{n,R}^N>c\varepsilon,\,
\mathbb{E}\left[
A_{n,R}^N
\mid\mathcal{F}_{n-1}^N
\right]\leq \eta
\right).
\end{align}
We start by bounding the first term. Conditional on
$\mathcal{F}_{n-1}^N$, the random variables $X_n^1,\ldots,X_n^N$
are i.i.d.\ with density $q_n^N$. Hence
\begin{align*}
\mathbb{E}\left[
A_{n,R}^N
\mid\mathcal{F}_{n-1}^N
\right]
&=
\int_{\{\|x\|>R\}}
\|x\|^2w_n(x)q_n^N(x)\,dx
=
\int_{\{\|x\|>R\}}
\|x\|^2
\pi(x)^{1-e^{-\gamma}}
(q_n^N(x))^{e^{-\gamma}}\,dx.
\end{align*}
By the weighted arithmetic--geometric mean inequality,
\(
\pi^{1-e^{-\gamma}}
(q_n^N)^{e^{-\gamma}}
\leq
(1-e^{-\gamma})\pi+e^{-\gamma}q_n^N
\) for each $x\in\R^d$.
Therefore, since $\pi\in\mathcal{P}_2(\mathbb{R}^d)$ and $\int_{\{\|x\|>R\}}
\|x\|^2 \pi(x)dx$ converges to zero as $R\to\infty$, we can choose $R$ large enough such that 
\begin{align*}
(1-e^{-\gamma})\int_{\{\|x\|>R\}}
\|x\|^2 \pi(x)dx <\frac{\eta}{2}
\end{align*}
from which we get
\begin{align}
\label{eq:Rfix}
\P\left(\mathbb{E}\left[
A_{n,R}^N
\mid\mathcal{F}_{n-1}^N
\right]>\eta\right)
&\leq
\P\left(\int_{\{\|x\|>R\}}
\|x\|^2 q_n^N(x)dx>\frac{\eta}{2} e^{\gamma} \right).
\end{align}

Recall that $q_n^N(x) = R_\gamma\mu_{n-1}^N(x)$ with $R_\gamma$ as in~\eqref{eq:Rgamma}, then
\begin{align*}
Q_{n,R}^N
&:=
\int_{\{\|x\|>R\}}
\|x\|^2 q_n^N(x)\,dx
=
\int_{\{\|x\|>R\}}
\|x\|^2
\int_{\mathbb{R}^d}
R_\gamma(y,dx)\,\mu_{n-1}^N(dy)
\\
&=
\int_{\mathbb{R}^d}
\mathbb{E}\left[
\|y+\gamma\nabla\log\pi(y)+\sqrt{2\gamma}Z\|^2
\mathbf{1}_{\{
\|y+\gamma\nabla\log\pi(y)+\sqrt{2\gamma}Z\|>R
\}}
\right]
\mu_{n-1}^N(dy),
\end{align*}
where $Z\sim\mathcal{N}(0,\textsf{Id})$.
Using the Lipschitz continuity of $\nabla\log\pi$, there exists
$C<\infty$ such that
\(
\|y+\gamma\nabla\log\pi(y)+\sqrt{2\gamma}Z\|
\leq
C(1+\|y\|+\|Z\|),
\) and $$\|y+\gamma\nabla\log\pi(y)+\sqrt{2\gamma}Z\|^2
\leq
3C^2(1+\|y\|^2+\|Z\|^2)\leq 3C^2(1+\|y\|^2)(1+\|Z\|^2)$$ since $(1+a+b)^2\leq 3(1+a^2+b^2)$.
Therefore, for $R\geq 2C$ ,
\[
\left\{
\|y+\gamma\nabla\log\pi(y)+\sqrt{2\gamma}Z\|>R
\right\}
\subseteq
\left\{
\|y\|>\frac{R}{4C}
\right\}
\cup
\left\{
\|Z\|>\frac{R}{4C}
\right\},
\]
and $\mathbf{1}_{\{
\|y+\gamma\nabla\log\pi(y)+\sqrt{2\gamma}Z\|>R
\}} \leq \mathbf{1}_{\{\|y\|>R/(4C)\}}+\mathbf{1}_{\{\|Z\|>R/(4C)\}}$.
Thus, using $\|y\|>R/(4C)$ and that $\E[\|Z\|^2]=d$, we have
\begin{align*}
Q_{n,R}^N
&\leq
3C^2(d+1)\left(1+\frac{16C^2}{R^2}\right)\int_{\{\|y\|>R/(4C)\}}
\|y\|^2\mu_{n-1}^N(dy)
\\
&\quad+
3C^2\mathbb{E}\left[(1+\|Z\|^2)\mathbf{1}_{\{\|Z\|>\frac{R}{4C}\}}\right]
\int_{\mathbb{R}^d}
(1+\|y\|^2)\,\mu_{n-1}^N(dy):= T_{1,R}^N + T_{2,R}^N.
\end{align*}
We control the two terms separately.
For the first term, fix $\eta,\delta>0$ and set
\(
\rho
=
\eta e^\gamma/
(12C^2(d+1)(1+4C^2)).
\)
By~\eqref{ass:UI_P} at step $n-1$, there exist $R_1<\infty$ and $N_2\in\mathbb{N}$ such that, for all $N\geq N_2$,
\[
\mathbb{P}\left(
\int_{\{\|y\|>R_1\}}
\|y\|^2\,\mu_{n-1}^N(dy)>\rho
\right)
<
\frac{\delta}{8}.
\]
Choose $R$ sufficiently large such that $R\geq 2$ and
$R/(4C)\geq R_1$. Since
\(
1+16C^2/R^2\leq 1+4C^2,
\)
we have
\[
\begin{aligned}
\mathbb{P}\left(
T_{1, R}^N
>\frac{\eta}{4}e^\gamma
\right)
&\leq
\mathbb{P}\left(
\int_{\{\|y\|>R/(4C)\}}
\|y\|^2\,\mu_{n-1}^N(dy)>\rho 
\right)\leq
\mathbb{P}\left(
\int_{\{\|y\|>R_1\}}
\|y\|^2\,\mu_{n-1}^N(dy)>\rho
\right)
<
\frac{\delta}{8},
\end{aligned}
\]
for all $N\geq N_2$.

For the second term, the dominated convergence theorem gives $\mathbb{E}\left[(1+\|Z\|^2)\mathbf{1}_{\{\|Z\|>\frac{R}{4C}\}}\right]\to 0$ as $R\to \infty$. 
Fix \(\eta,\delta>0\). By~\eqref{eq:first_step}, applied with \(\rho=\delta/8\), there exist \(M<\infty\) and \(N_3\in\mathbb N\) such that, for all \(N\geq N_3\),
\[
\mathbb{P}\left(
\int_{\mathbb R^d}\|y\|^2\,\mu_{n-1}^N(dy)>M
\right)
<
\frac{\delta}{8}.
\]
Choose \(R\) sufficiently large that
\begin{equation}
\label{eq:R_cond_2}
3C^2(1+M)\mathbb{E}\left[(1+\|Z\|^2)\mathbf{1}_{\{\|Z\|>\frac{R}{4C}\}}\right]<\frac{\eta}{4}e^\gamma.
\end{equation}
On the event
\(
\left\{
\int_{\mathbb R^d}\|y\|^2\,\mu_{n-1}^N(dy)\leq M
\right\},
\)
we have
\(
\int_{\mathbb R^d}(1+\|y\|^2)\,\mu_{n-1}^N(dy)
\leq 1+M.
\)
Consequently, for all \(N\geq N_3\)
\[
\begin{aligned}
&\mathbb{P}\left(T_{2, R}^N>\frac{\eta}{4}e^\gamma
\right)\leq
\mathbb{P}\left(
\int_{\mathbb R^d}\|y\|^2\,\mu_{n-1}^N(dy)>M
\right)
<
\frac{\delta}{8}
\end{aligned}
\]
This shows that for $N\geq \max\{N_2, N_3\}$ and $R\geq \max\{2, 2C, 4CR_1\}$ and such that~\eqref{eq:R_cond_2} holds we have
\begin{align}
\label{eq:QN_control}
\mathbb{P}\left(Q_{n, R}^N>\frac{\eta}{2}e^\gamma\right)<\frac{\delta}{4}.
\end{align}
Then, the conditional Markov's inequality shows that the second term in~\eqref{eq:AN_decomposition} is bounded
\[
\begin{aligned}
&\mathbb{P}\left(
A_{n,R}^N>c\varepsilon,\,
\mathbb{E}\left[
A_{n,R}^N
\mid\mathcal{F}_{n-1}^N
\right]\leq \eta
\right)\leq
\frac{\eta}{c\varepsilon}=\frac{\delta}{4}.
\end{aligned}
\]

Combining~\eqref{eq:proba_Z}, \eqref{eq:AN_decomposition}, \eqref{eq:Rfix} and~\eqref{eq:QN_control} we get that for sufficiently large $N, R$
\begin{align*}
\mathbb{P}\left(
\int_{\{\|x\|>R\}}
\|x\|^2\,\mu_n^N(dx)
>\varepsilon
\right)
<
\delta.
\end{align*}
This proves~\eqref{ass:UI_P} for $\mu_n^N$ and completes the
induction.

\end{proof}

\section{Comparison with other Monte Carlo algorithms}
\label{app:targets}




\textbf{2D Gaussian mixture (Figure~\ref{fig:gm2d})}
The 2D 4 component Gaussian mixture  is given by $\pi(x) = \sum_{i=1}^4 w_i \mathcal{N}(x; \mathbf{m}_i, \Sigma_i)$
with 
\begin{align*}
    w &= \left[\frac{1}{3}, \frac{1}{6}, \frac{1}{3}, \frac{1}{6} \right]^\top; \qquad\qquad
    \mathbf{m}_1 = [3, 3]^\top; \quad\mathbf{m}_2 = [9, 10]^\top; \quad\mathbf{m}_3 = [2, 2]^\top;  \quad\mathbf{m}_4 = [10, 4]^\top \\
    \Sigma_1 &= 0.7 \cdot  \left[ \begin{matrix}
        1 & -0.8377 \\
        -0.8377 & 1
    \end{matrix}  \right]; \quad\Sigma_2 = 0.7 \cdot  \left[ \begin{matrix}
        1 & 0.8377 \\
        0.8377 & 1
    \end{matrix}  \right]; \quad
    \Sigma_3 = \Sigma_1; \quad
    \Sigma_4 =  2 \cdot  \left[ \begin{matrix}
        1 & -0.6701 \\
        -0.6701 & 1
    \end{matrix}  \right].
\end{align*}
For the 2D Gaussian mixture target we use $N=300, \gamma = 0.2$. We tune the variance of the proposal for MALA and SMC-MALA to achieve $57\%$ acceptance rate and the variance of RWM for SMC tempering to achieve $23\%$ acceptance rate.
We use $T=1000$ iterations for SMC-WFR.
The number of iterations of the other algorithms is selected so that the runtime of all algorithms is comparable.
For the 2D Gaussian Mixture ULA does not work well and fails to explore all modes.  A similar issue affects MALA and SMC-MALA. SMC-ULA gives the worst results as it combines the lack of exploration of ULA with the instability in the weights~\eqref{eq:w_smcula} due to the lack of overlap between $\mu_0$ and $\pi$.
SMC-tempering works well in this setting, as the tempering iterates~\eqref{eq:tempering} flatten the target in the early iterations and help with global exploration. SMC-WFR combines the local exploration of ULA with the ability to jump from one mode to another and has the fastest convergence, although it has a larger bias than SMC-tempering due to the large value of $\gamma$ selected. On the other hand, the fast convergence of SMC-WFR occurs only for sufficiently large $\gamma$.

\textbf{Banana (Figure~\ref{fig:banana_app})}
For the banana target we have $\pi(x)\propto \exp\left(-[\sigma^2(x_2-x_1^2)^2/2+\nicefrac{(1-x_1)^2}{2\sigma^2}]\right)$ with $\sigma^2=2$. We use $N=300, \gamma = 0.005$. We tune the variance of the proposal for MALA and SMC-MALA to achieve $57\%$ acceptance rate and the variance of RWM for SMC tempering to achieve $23\%$ acceptance rate.
We use $T=4000$ iterations for SMC-WFR.
The number of iterations of the other algorithms is selected so that the runtime of all algorithms is comparable.
For the banana target ULA and SMC-ULA are slower to convergence as in the initial stages the drift component of ULA pushes the particles far in the tails. SMC-WFR does not suffer from this instability despite the ULA proposal. This experiment showcases the benefit of using the weights~\eqref{eq:w_wfr} rather than~\eqref{eq:w_smcula} despite the additional computational cost.
MALA and SMC-MALA provide the best results.
SMC-tempering has similar convergence speed but tends to underestimate the covariance matrix.

\begin{figure}
    \centering
    \includegraphics[width=\linewidth]{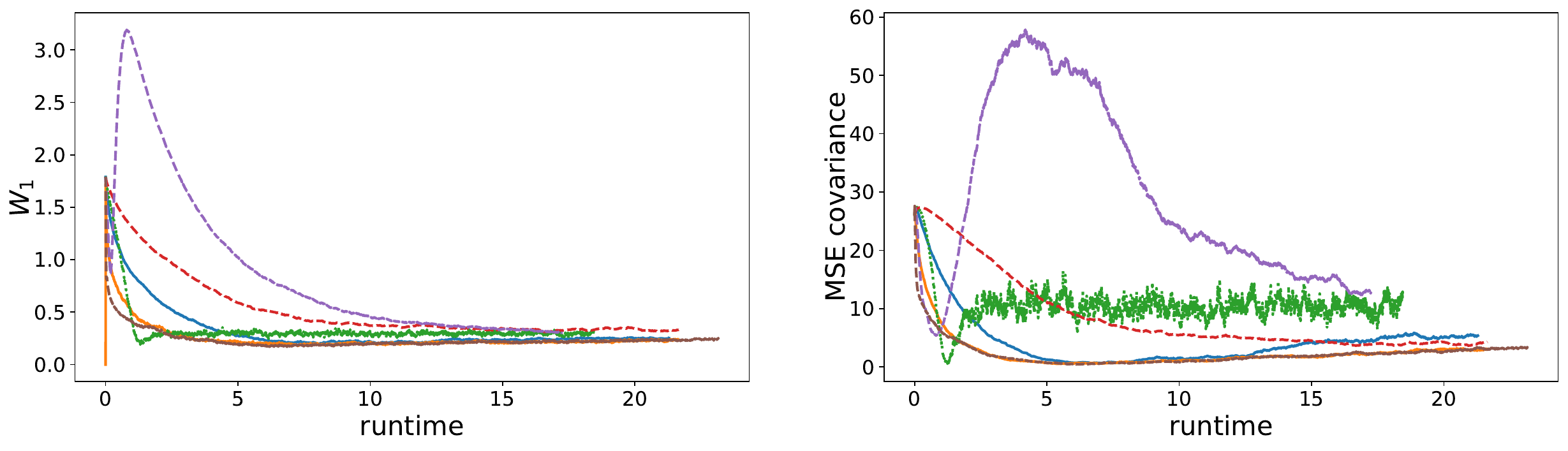}
    \includegraphics[width=\linewidth]{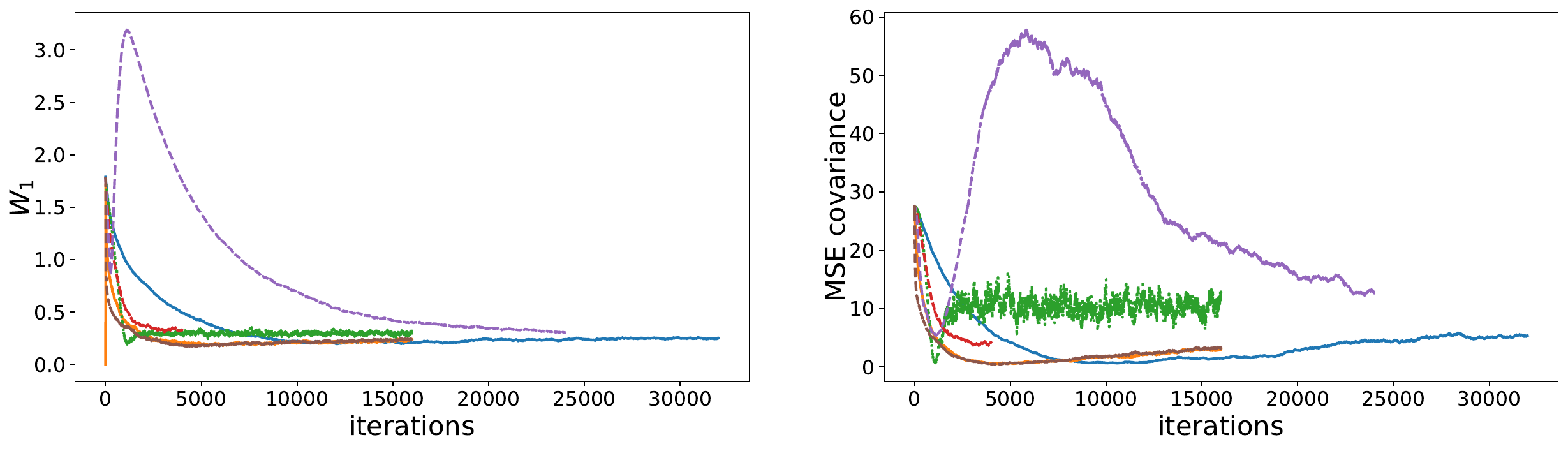}
    \includegraphics[width = 0.8\textwidth]{legend_comparison.pdf}
    \caption{Banana: Evolution of $W_1$ and mean squared error (MSE) for the covariance matrix against runtime (first row) and number of iterations (second row) averaged over 50 repetitions.  The initial distribution is $\mu_0(x) = \mathcal{N}(x; 0, \textsf{Id})$. We compare ULA, MALA, SMC tempering, SMC-WFR, SMC-ULA and SMC-MALA.}
    \label{fig:banana_app}
\end{figure}

\textbf{Donut (Figure~\ref{fig:donut_app})}
Let $\pi(x)\propto \exp\left(-(\sqrt{x_1^2+x_2^2}-2)^2/0.25^2\right)$.
For the donut target we use $N=300, \gamma = 0.005$. 
We tune the variance of the proposal for MALA and SMC-MALA to achieve $57\%$ acceptance rate and the variance of RWM for SMC tempering to achieve $23\%$ acceptance rate. 
We use $T=1000$ iterations for SMC-WFR.
The number of iterations of the other algorithms is selected so that the runtime of all algorithms is comparable.
For the donut target ULA works well while SMC-tempering struggles to quickly move from $\mu_0$ to $\pi$ since it does not use gradient information. SMC-WFR and SMC-ULA improve the results of ULA as guaranteed by the converge results of WFR~\eqref{eq:rate_wfr}. For this target SMC-ULA and SMC-WFR behave similarly, we do not observe any numerical instability in SMC-ULA as $\mu_0$ significantly overlaps with $\pi$ and the weights~\eqref{eq:w_smcula} are stable. 
SMC-tempering is slow to converge due to the lack of gradient information.
\begin{figure}
    \centering
     \includegraphics[width=\linewidth]{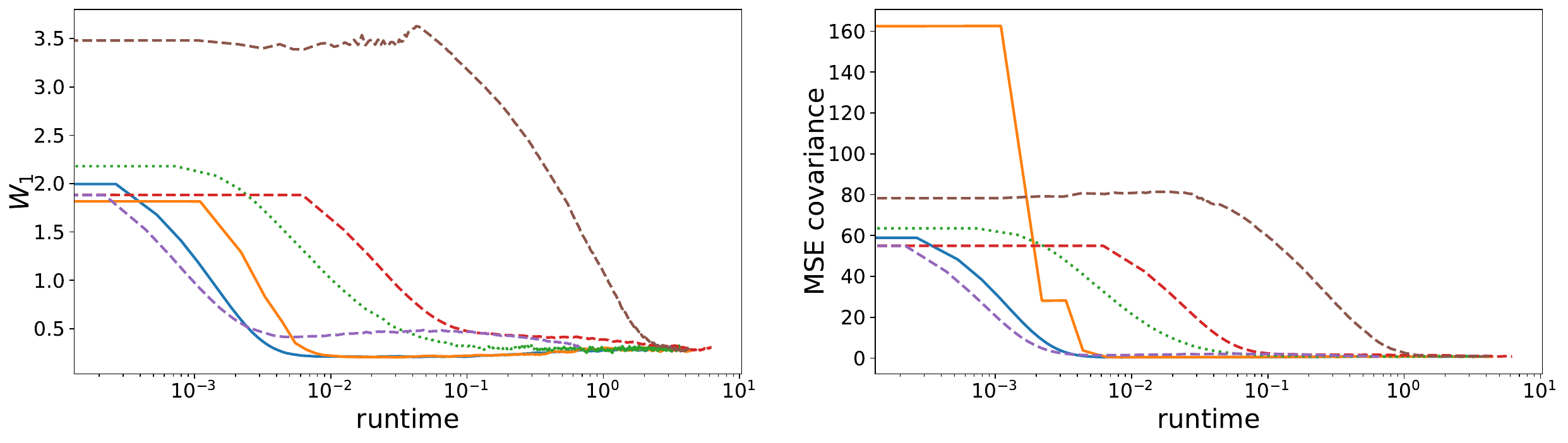}
     \includegraphics[width=\linewidth]{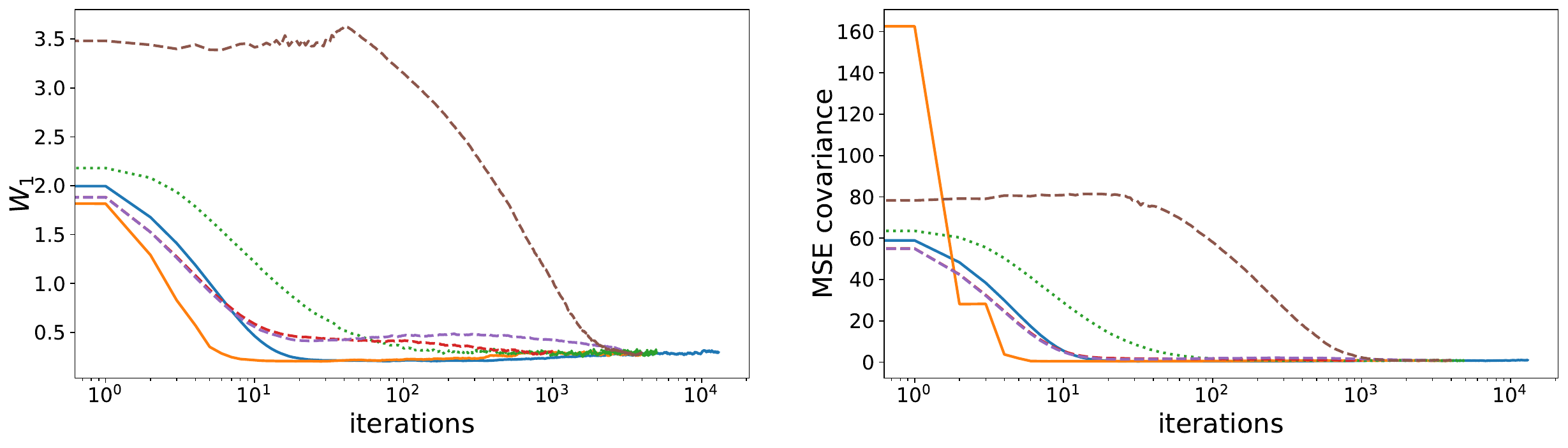}
     \includegraphics[width = 0.8\textwidth]{legend_comparison.pdf}
    \caption{Donut (ULA works well): Evolution of $W_1$ and mean squared error (MSE) for the covariance matrix against runtime (first row) and number of iterations (second row) averaged over 50 repetitions.  The initial distribution is $\mu_0(x) = \mathcal{N}(x; 0, \textsf{Id})$. We compare ULA, MALA, SMC tempering, SMC-WFR, SMC-ULA and SMC-MALA. ULA, SMC-ULA, SMC-WFR are similar in terms of number of iterations but ULA is faster in terms of wallclock time.}
    \label{fig:donut_app}
\end{figure}

\textbf{20D Gaussian (Figure~\ref{fig:20DG1} and~\ref{fig:20DG2})}
We consider two 20D Gaussian distributions (thus $\nabla V_\pi$ is globally Lipschitz continuous). One in which the W flow performs well since the log-Sobolev constant is small (and so $C_{\textrm{LSI}}^{-1}$ is large) and a second one  in which the log-Sobolev constant is large and thus the W flow performs less well.
The first target has 
     $\mathbf{m} = 3\cdot \mathsf{1}_{20}, 
    \Sigma = 2 \cdot C$
where $\mathsf{1}_{20}$ is the 20D one vector.  The correlation matrix $C$ can be found in the supplementary material, in this case $C_{\textrm{LSI}}=0.077$.
The second target has 
    $\mathbf{m} = 20\cdot  \mathsf{1}_{20},
    \Sigma = 5 \cdot\mathsf{Id}_{20 \times 20}$
where $\mathsf{Id}_{20 \times 20}$ is the 20D identity matrix. In this case $C_{\textrm{LSI}}=5$. 
For both targets we use $N=1000$ particles and $T=100$ iterations for SMC-WFR.
The number of iterations of the other algorithms is selected so that the runtime of all algorithms is comparable.
We set $\gamma = 0.05$ for target 1 and $\gamma = 0.5$ for target 2.
We tune the variance of the proposal for MALA and SMC-MALA to achieve $57\%$ acceptance rate and the variance of RWM for SMC tempering to achieve $23\%$ acceptance rate. 
For the first 20D Gaussian target (small log-Sobolev constant) the theory suggests that the W flow is faster than FR. This is confirmed by Figure~\ref{fig:20DG1} which shows SMC-WFR, SMC-ULA and SMC-MALA outperforming both ULA and SMC-tempering on this target in terms of iterations. However, ULA works well and is cheap compared to alternatives and is the faster than SMC-WFR to convergence when taking runtime into account. 
In this case, SMC-ULA seems to offer the best of both worlds by closely following the WFR flow (approximated by SMC-WFR) and being computationally cheap.
 \begin{figure}
     \centering
      \includegraphics[width=\linewidth]{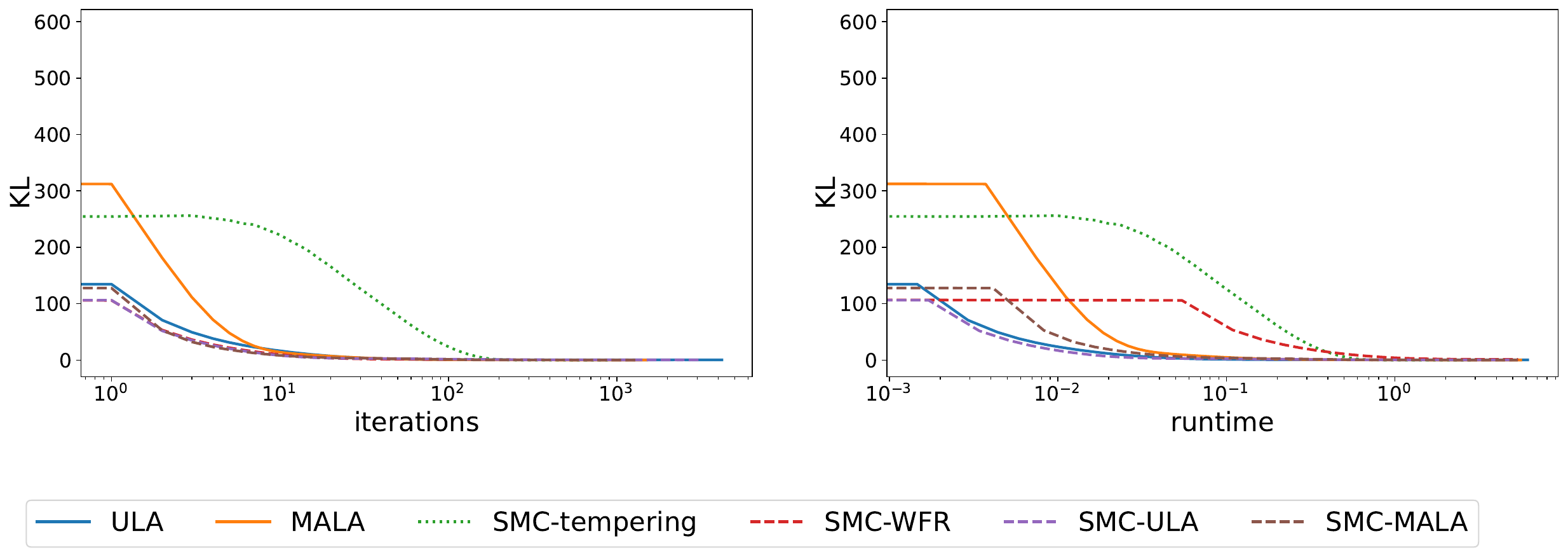}
     \caption{20D Gaussian with small log-Sobolev constant (ULA works well): Evolution of KL against runtime (right) and number of iterations (left) averaged over 50 repetitions.  The initial distribution is $\mu_0(x) = \mathcal{N}(x; 0, \textsf{Id})$. We compare ULA, MALA, SMC tempering, SMC-WFR, SMC-ULA and SMC-MALA. SMC-WFR and SMC-ULA achieve the fastest convergence in number of iterations but SMC-ULA and ULA are faster than SMC-WFR in terms of runtime due to their lower computational cost.}
     \label{fig:20DG1}
 \end{figure}
For the second 20D Gaussian target (large log-Sobolev constant) ULA is less fast and in fact MALA based algorithms tend to perform better compared to ULA based ones. SMC-ULA and SMC-WFR perform similarly in terms of iteration number suggesting that in this case SMC-ULA is a valid alternative to approximate the WFR flow.

  \begin{figure}
     \centering
      \includegraphics[width=\linewidth]{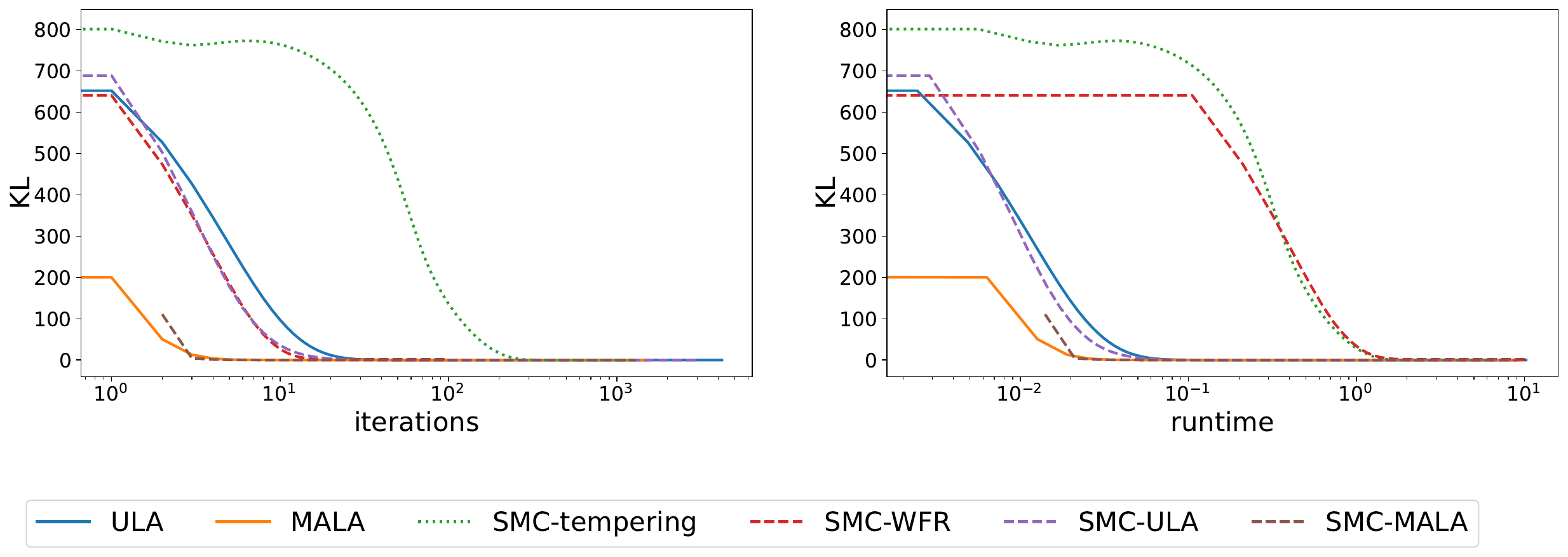}
     \caption{20D Gaussian with large log-Sobolev constant (ULA struggles): Evolution of KL against runtime (right) and number of iterations (left) averaged over 50 repetitions.  The initial distribution is $\mu_0(x) = \mathcal{N}(x; 0, \textsf{Id})$. We compare ULA, MALA, SMC tempering, SMC-WFR, SMC-ULA and SMC-MALA. SMC-WFR and SMC-ULA are better than ULA in terms of iteration number. MALA and SMC-MALA outperform ULA based algorithms thanks to tuning of proposal variance.}
     \label{fig:20DG2}
 \end{figure}

\textbf{20D Gaussian mixture (Figure~\ref{fig:gm20d})}
The 20D 3 component Gaussian mixture  is given by $\pi(x) = \sum_{i=1}^3 w_i \mathcal{N}(x; \mathbf{m}_i, \Sigma_i)$
with 
\begin{align*}
    w &= \left[\frac{1}{2}, \frac{1}{3}, \frac{1}{6} \right]^\top; \qquad\qquad
    \mathbf{m}_1 = 3\cdot  \mathsf{1}_{20}; \quad\mathbf{m}_2 = 4\cdot  \mathsf{1}_{20};  \quad\mathbf{m}_3 = [2, 5, \dots, 2, 5]^\top \\
    \Sigma_1 &= 2 \cdot  C; \quad\Sigma_2 = 2 \cdot C; \quad
    \Sigma_3 = C^{-1}M^{-1}
\end{align*}
where $\mathsf{1}_{20}$ is the 20D one vector.  The correlation matrices $C, M$ can be found in the supplementary material.
For the 20D Gaussian mixture target we use $N=1000, \gamma = 0.1$. We tune the variance of the proposal for MALA and SMC-MALA to achieve $57\%$ acceptance rate and the variance of RWM for SMC tempering to achieve $23\%$ acceptance rate.
We use $T=1000$ iterations for SMC-WFR.
The number of iterations of the other algorithms is selected so that the runtime of all algorithms is comparable.
For the 20D Gaussian Mixture ULA, MALA, SMC-ULA and SMC-MALA do not work well due to the different covariance structure of the 3 components.
SMC-tempering also fails on this target. SMC-WFR is the only algorithm to correctly recover the covariance structure of the target.